\documentclass{article}

\usepackage{fullpage}
\usepackage{multirow}
\usepackage{tikz}
\usepackage[utf8]{inputenc} 
\usepackage[T1]{fontenc}    
\usepackage{hyperref}       
\usepackage{url}            
\usepackage{booktabs}       
\usepackage{amsfonts}       
\usepackage{nicefrac}       
\usepackage{microtype}      
\usepackage{xcolor}         
\usepackage{comment}

\usepackage{amsmath,amssymb,amsfonts, amsthm}
\usepackage{dsfont}
\usepackage{xspace}
\usepackage{thm-restate}
\usepackage{algorithm}
\usepackage{algorithmic}
\usepackage{graphicx}
\usepackage{subcaption}
\usepackage{url}

\newtheorem{theorem}{Theorem}[section]
\newtheorem{definition}[theorem]{Definition}
\newtheorem{corollary}[theorem]{Corollary}
\newtheorem{remark}[theorem]{Remark}
\newtheorem{lemma}[theorem]{Lemma}
\newtheorem{observation}[theorem]{Observation}

\DeclareMathOperator*{\argmin}{argmin}

\DeclareMathOperator*{\poly}{poly}

\makeatletter
\renewcommand*{\@fnsymbol}[1]{\textcolor{blue}{\ensuremath{\ifcase#1\or *\or \dagger\or \ddagger\or
 \mathsection\or \triangledown\or \mathparagraph\or \|\or **\or \dagger\dagger
   \or \ddagger\ddagger \else\@ctrerr\fi}}}
\makeatother

\providecommand{\email}[1]{\href{mailto:#1}{\nolinkurl{#1}\xspace}}

\newcommand{\blockcomment}[1]{}

\definecolor{darkblue}{rgb}{0.0, 0.0, 0.55}
\definecolor{bleudefrance}{rgb}{0.19, 0.55, 0.91}
\definecolor{richmaroon}{rgb}{0.69, 0.19, 0.38}
\hypersetup{
     colorlinks   = true,
     citecolor    = bleudefrance,
  	 linkcolor	  = bleudefrance
}

\providecommand{\email}[1]{\href{mailto:#1}{\nolinkurl{#1}\xspace}}

\newcommand{\PPr}[1]{\ensuremath{\mathbf{Pr}\left[#1\right]}}
\newcommand{\Ex}[1]{\ensuremath{\mathbb{E}\left[#1\right]}}

\newcommand{\eps}{\varepsilon}
\newcommand{\C}{\textup{cost}}
\newcommand{\mdef}[1]{{\ensuremath{#1}}\xspace}
\def \bA    {\mdef{\textbf{A}}}
\def \bB    {\mdef{\textbf{B}}}
\def \bX    {\mdef{\textbf{X}}}
\def \bP    {\mdef{\textbf{P}}}

\def \calC    {\mdef{\mathcal{C}}}

\def \calE    {\mdef{\mathcal{E}}}

\def \R    {\mdef{\mathbb{R}}}

\newcommand{\E}{\mathbb{E}}
\newcommand{\I}{\mathds{1}}

\newcommand{\phit}{\widetilde{\varphi}}

\renewcommand{\d}{\delta}
\renewcommand{\v}{\nu}

\title{Dimensionality Reduction for Wasserstein Barycenter}

\author{
Zachary Izzo
\\Stanford University
\thanks{E-mail: \email{zle.izzo@gmail.com}}
\and
Sandeep Silwal
\\MIT\thanks{E-mail: \email{silwal@mit.edu}}
\and
Samson Zhou
\\Carnegie Mellon University\thanks{E-mail: \email{samsonzhou@gmail.com}}
}

\begin{document}

\maketitle

\begin{abstract}
The Wasserstein barycenter is a geometric construct which captures the notion of centrality among probability distributions, and which has found many applications in machine learning. However, most algorithms for finding even an approximate barycenter suffer an exponential dependence on the dimension $d$ of the underlying space of the distributions. In order to cope with this ``curse of dimensionality,'' we study dimensionality reduction techniques for the Wasserstein barycenter problem. When the barycenter is restricted to support of size $n$, we show that randomized dimensionality reduction can be used to map the problem to a space of dimension $O(\log n)$ independent of both $d$ and $k$, and that \emph{any} solution found in the reduced dimension will have its cost preserved up to arbitrary small error in the original space. We provide matching upper and lower bounds on the size of the reduced dimension, showing that our methods are optimal up to constant factors. We also provide a coreset construction for the Wasserstein barycenter problem that significantly decreases the number of input distributions. The coresets can be used in conjunction with random projections and thus further improve computation time. Lastly, our experimental results validate the speedup provided by dimensionality reduction while maintaining solution quality.
\end{abstract}

\section{Introduction} \label{sec:intro}
The Wasserstein barycenter (WB) is a popular method in statistics and machine learning for summarizing data from multiple sources while capturing their underlying geometry~\cite{martial11}.
The problem is defined as follows. 
Suppose we have a collection of data, represented as $k$ discrete probability distributions $\mu_1,\ldots,\mu_k$ on $\mathbb{R}^d$. Given a set of non-negative weights $\lambda_1,\ldots,\lambda_k$ that sum to $1$, and a class $\mathbb{P}$ of probability distributions on $\mathbb{R}^d$, a Wasserstein barycenter under the $L_p$ objective for a parameter $p>0$ is a probability distribution $\nu\in\mathbb{P}$ that minimizes
\begin{equation} \label{eq:bary def}
    \sum_{i=1}^k\lambda_iW_p(\mu_i,\nu)^p,
\end{equation}
where $W_p(\mu_i,\nu)$ is the $p$-Wasserstein distance. 

The Wasserstein barycenter is a natural quantity that captures the geometric notion of centrality among point clouds, as it utilizes the optimal transport distance~\cite{bertsimas1997introduction} between a number of observed sets. 
Thus, Wasserstein barycenters have been extensively used in machine learning~\cite{srivastava2018scalable}, data sciences~\cite{ruschendorf2002n,elvander2020multi}, image processing~\cite{rubner1997earth}, computer graphics~\cite{pele2009fast}, and statistics~\cite{villani2008optimal}, with applications in constrained clustering~\cite{cuturi2014fast, ho2017multilevel}, Bayesian learning~\cite{srivastava2018scalable}, texture mixing~\cite{rabin2011wasserstein}, and shape interpolation~\cite{solomon2015convolutional}. 

Unfortunately, the problem is NP-hard to compute \cite{AltschulerB21, BorgwardtP21} and many algorithms that even approximate the Wasserstein barycenter suffer from large running times, especially if the datasets are high dimensional~\cite{MuzellecC19}. 
Indeed, \cite{altschuler2021wasserstein} recently gave an algorithm that computes the Wasserstein barycenter using runtime that depends exponentially on the dimension, thus suffering the ``curse of dimensionality.''

To alleviate these computational constraints, we consider dimensionality reduction for computing the Wasserstein barycenter. 
Dimensionality reduction can be used to improve the performance of downstream algorithms on high dimensional datasets in many settings of interest, e.g., see the survey~\cite{cunningham2015linear}. 
In the specific case of Wasserstein barycenters, dimensionality reduction has several practical and theoretical benefits, including lower storage space, faster running time in computing distances, and versatility: it can be used as a pre-processing tool and combined with \emph{any} algorithm for computing the Wasserstein barycenter.

\subsection{Our Results}
In this paper, we study dimensionality reduction techniques for computing a Wasserstein barycenter of discrete probability distributions.
Our main results show that it is possible to project the distributions into low dimensions while provably preserving the quality of the barycenter. 
A key result in dimensionality reduction is the classical Johnson-Lindenstrauss (JL) lemma~\cite{johnson1984extensions}, which states that projecting a dataset of $N$ points into roughly $O(\log N)$ dimensions is enough to preserve all pairwise distances. 

Using the JL lemma, we first show that we can assume the distributions lie in $O(\log(nk))$ dimensions, where $k$ is the number of input distributions whose barycenter we are computing, $n$ is the size of the support of the barycenter, and each of the $k$ input distributions has support size $\poly(n)$. 
For $p=2$, there exists a closed form for the cost of any candidate barycenter in terms of the pairwise distances of the points in the input distributions. 
Thus it is straightforward to see that our bound results from the fact that there are $k \cdot \poly(n)$ total points masses in the union of all the distributions and therefore, projecting them into a dimension of size $O(\log( k \poly(n))) = O(\log(nk))$ suffices to preserve all of their pairwise distances.
However for $p\neq 2$, a closed form for the optimal cost no longer exists, so preservation of all pairwise distances is insufficient. 
Instead, we make use of a Lipschitz extension theorem, namely the Kirszbraun theorem, which allows us to ``invert'' the dimensionality reduction map and argue the preservation of the cost of the Wasserstein barycenter under a general $L_p$ objective. For more details, see Section~\ref{sec:log(nk) reduction}.

\textbf{Dimensionality reduction independent of $k$.} 
While the JL lemma is known to be tight~\cite{LarsenN16, LarsenN17}, it is possible to improve its dimensionality guarantees for \emph{specific} problems, such as various formulations of clustering \cite{CohenEMMP15,BecchettiBC0S19,MakarychevMR19}. 
Indeed, our main result is that we can achieve a dimension bound \emph{beyond} the $O(\log(nk))$ bound that follows from the JL lemma and Kirszbraun theorem.  
We show that it suffices to project the support points onto $O(\log n)$ dimensions, which is \emph{independent} of the number of distributions $k$. 
In fact, we show a stronger statement that projecting the points supported by the distributions onto $O(\log n)$ dimension preserves the cost of the objective \eqref{eq:bary def} for \emph{any} distribution $\nu$ supported on at most $n$ points (Theorem~\ref{thm:jl:main}). 
The algorithmic application of this theorem is that one can take \emph{any} approximation algorithm or heuristic for computing the Wasserstein barycenter and combine it with dimensionality reduction. 
A simplification of our theorem is stated below where we omit some parameters for clarity.

\begin{theorem}[Theorem \ref{thm:jl:main} Simplified] \label{thm:jl:main:simple}
Let $\mu_1,\ldots,\mu_k$ be discrete probability distributions on $\R^d$ such that $|\textrm{supp}(\mu_i)| \le \poly(n)$ for all $i$. There exists a dimensionality reduction map $\pi:\mathbb{R}^d\to\mathbb{R}^m$ for $m=O(\log n)$ such that projection under $\pi$ preserves the cost of objective \eqref{eq:bary def} for any $\nu$ supported on at most $n$ points.
\end{theorem}

The result is surprising because the projected dimension is independent of the number of input distributions $k$, which could be significantly larger than $n$. 
Thus the random projection map $\pi$ can no longer even guarantee the preservation of a significant fraction of pairwise distances between the support points of the $k$ distributions. 
Our main tool is a ``robust'' Lipschitz extension theorem introduced in \cite{MakarychevMR19} for $k$-means clustering. 
We adapt this analysis to the geometry of the Wasserstein barycenter problem. 

\textbf{Optimality of dimensionality reduction.} We complement our upper bound results by showing that our dimension bound of $\log n$ dimensions is \emph{tight} if a random Gaussian matrix is used as the projection map. 
We also show that the JL lemma is tight for the related problem of computing the optimal transport between two distributions with support of size $n$. 
More specifically, we give a lower bound showing that $\Omega(\log n)$ dimension is needed for a random projection to preserve the optimal transport cost. 
Thus our results show a separation between the geometry of the optimal transport problem and the geometry of the Wasserstein barycenter problem, as we overcome the JL bound in the latter. 

\textbf{Hardness of approximation.} 
In addition, we also show the NP-hardness of approximation for the Wasserstein barycenter problem. 
Namely, we show that it is NP-hard to find an approximate barycenter that induces a cost that is within a factor of $1.0013$ of the optimal barycenter if we restrict the support size of the barycenter. 
This complements recent work of \cite{AltschulerB21, BorgwardtP21}, who showed that computing sparse Wasserstein barycenters is NP-hard. 

\textbf{Coresets for Wasserstein barycenters.}
An alternate way to reduce the complexity of datasets is through the use of coresets, which decrease the effective data size by reducing the number of input points rather than the input dimension $d$. 
If the number of input distributions $k$ is significantly larger than the support size $n$, we show that there exists a weighted subset $C$ of roughly $\poly(n)$ distributions, so that computing the optimal barycenter on $C$ is equivalent to computing the optimal barycenter on the original input up to a small approximation loss. 
Hence, it can potentially be much more efficient to use the subset $C$ in downstream algorithms involving Wasserstein barycenters. 
Moreover, the coreset is not mutually exclusive with our techniques for reducing the ambient dimension $d$. 
Our techniques show that we can simultaneously reduce both the size of the input distribution $k$ and the dimension $d$ of the data, while preserving the optimal clustering within a small approximation factor.

In Supplementary Section \ref{sec:sup_constrained}, we also show a connection between the Wasserstein barycenter problem and constrained low-rank problems. 
This class of problems includes examples such as the singular value decomposition (SVD) and $k$-means clustering. 
While this connection does not yield any improved results, it classifies the Wasserstein barycenter as a member of a general class of problems, and this classification could have further applications in the future.

\textbf{Experiments.}
Finally, we present experimental evaluation of our proposed methodology. Note that our results imply that we can use dimensionality reduction in conjunction with \emph{any} Wasserstein barycenter algorithm and still roughly retain the approximation guarantees of the algorithm used. Specifically, we give examples of real high dimensional datasets such that solving the Wasserstein barycenter problem in a reduced dimension leads to computational savings while preserving the quality of the solution. Our experiments in Section \ref{sec:experiments} demonstrate that on natural datasets, we can reduce the dimension by $1$-$2$ orders of magnitude while increasing the solution cost by only $5\%$. We also empirically test our coreset construction. Our method both reduces error and requires fewer samples than simple uniform sampling.

\subsection{Related Work}
\cite{AltschulerB21, BorgwardtP21} showed that computing sparse Wasserstein barycenters is NP-hard; hence, most of the algorithmic techniques focus on computing approximate Wasserstein barycenters that induce a cost within an additive $\eps$ of the optimal cost. 
\cite{agueh2011barycenters} first considered approximating Wasserstein barycenters when either (1) the distributions $\mathbb{P}$ only have discrete support on $\mathbb{R}$, (2) $k=2$, or (3) the distributions $\mu_i$ are all multivariate Gaussians in $\mathbb{R}^d$. 
Although there is a line of research that studies the computation of barycenters of continuous distributions, e.g.~\cite{alvarez2016fixed,chewi2020gradient}, we focus on discrete input distributions. For discrete input distributions, the majority of the literature can be categorized by its assumptions of the support of the barycenter~\cite{altschuler2021wasserstein}. 

\textbf{Fixed-support.} 
The ``fixed-support approximation'' class of algorithms assume that the support of the barycenter is among a fixed set of possible points. 
It then remains for the algorithms to solve a polynomial-size linear program associated with the corresponding set~\cite{cuturi2014fast,benamou2015iterative,carlier2015numerical,staib2017parallel,kroshnin2019complexity,lin2020fixed}. 
Unfortunately, the set of possible points must often be an $\varepsilon$-net over the entire space, which results in a size proportional to $1/\varepsilon^d$ that suffers from the curse of dimensionality. 
Nevertheless for constant dimension, the algorithms typically have runtime $\poly(n,k,D/\eps)$, where $D$ is an upper bound on the diameter of the supports of the input distributions. 
This is further improved by an algorithm of~\cite{altschuler2021wasserstein} that achieves runtime $\poly(n,k,\log(D/\eps))$. 

\textbf{Free support.} 
A separate class of algorithms do not make assumptions about the possible support of the optimal barycenter. 
These ``free-support algorithms'' instead optimize over the entire set of up  candidate barycenters, which can be as large as $n^k$ in quantity. 
Thus these algorithms, e.g.,~\cite{cuturi2014fast,LuiseSPC19}, either use exponential runtime or a number heuristics that lack theoretical guarantees. \cite{altschuler2021wasserstein} showed how to explore the $n^k$ possible points in polynomial time for fixed $d$.

\section{Preliminaries} \label{sec:prelims}
\textbf{Notation.}
For a positive integer $n$, we denote $[n]:=\{1,2,\ldots,n\}$. We use $\mu_1, \ldots, \mu_k$ to denote the $k$ distributions whose Wasserstein barycenter we wish to compute. While the Wasserstein barycenter problem is well defined for continuous distributions, in practice and in actual computations, the distributions $\mu_i$ are assumed to be discrete distributions that are supported on some number of point masses. This is also the assumption we make. More specifically, we assume that each of the distributions $\mu_i$ are discrete distributions supported on at most $T \le  n^C$ points where $C$ is a fixed constant. That is, $\mu_i = \sum_{j=1}^T a(x_{ij}) \d_{x_{ij}}$, where $\d_x$ is a delta function at $x$ and $a(x)$ is the weight assigned to a point $x$ in its corresponding $\mu_i$. We note that if there is some point $x$ in the support of more than one of the $\mu_i$s, then the weight function $a$ may not be well-defined. Instead, we implicitly assume that $a = a(x, i)$ is a function of both the point \emph{and} the distribution from which it comes, but we suppress this dependence on $i$ for notational clarity.

The distribution $\v$ denotes a candidate for the Wasserstein barycenter of the $\mu_i$. We write $\v = \sum_{j=1}^n b_j \d_{\v^j}$. In general, an actual Wasserstein barycenter (in the sense of minimizing the objective \eqref{eq:bary def} over all possible $\nu$ of any support size) may have support size up to $|\bigcup_{i=1}^k \textrm{supp}(\mu_i)|$ \cite{anderes16}. Throughout this paper, we will restrict ourselves to computing (approximate) barycenters of support size at most $n$. 
When we refer to an optimal barycenter, we mean a distribution that minimizes the objective \eqref{eq:bary def} \emph{within this restricted class}.

\textbf{Problem description.} 
The goal is to compute a distribution $\nu \in \mathbb{R}^d$, consisting of at most $n$ point masses, to minimize the objective \eqref{eq:bary def}.
As previously mentioned, $W_p(\mu_i, \nu)$ is the Wasserstein $p$-metric, defined as
$$W_p(\mu, \nu) = \inf_{\gamma \in \Gamma(\mu, \nu)} \left(\int_{\R^d \times \R^d} \| x - y \|^p d\gamma(x, y)\right)^{1/p}$$
where $\Gamma(\mu, \nu)$ is the set of all joint distributions with marginals $\mu$ and $\nu$ (i.e. all couplings of $\mu$ and $\nu$) and $\|\cdot\|$ denotes the Euclidean norm on $\R^d$. When $\mu$ and $\nu$ are discrete distributions, $W_p(\mu, \nu)^p$ $p$-metric is the cost of the minimum cost flow from $\mu$ to $\nu$ with edge costs being the Euclidean distance raised to the $p$-th power. For simplicity, we assume that the distributions $\mu_1, \cdots, \mu_k$ are weighted equally (each $\lambda_i = 1/k$ in \eqref{eq:bary def}) but our results hold in the general case as well. The most common choice of $p$ is $p=2$.

\textbf{Description of $\nu$.}  The barycenter $\nu$ can be characterized as follows. Recall that $\nu$ is supported on the points $\nu^1, \ldots, \nu^n$.
For the optimal coupling of each $\mu_i$ to $\v$, let $w_j(x)$ denote the total weight sent from $x$ (in the support of one of the $\mu_i$s) to $\v^j$. (The same note about suppressing the dependence of $w_j$ on the distribution $\mu_i$ from which $x$ comes applies here.)
Let $S_j = \{x \in \bigcup_{i=1}^k \mathrm{supp}(\mu_i) \: : \: w_j(x) > 0\}$ denote the set of all points in the $\mu_i$s with some weight sent to $\v^j$.
Then \emph{given} the set $S_j$ and weighting function $w_j(\cdot)$, we can reconstruct $\nu^j$ since it must minimize the objective 
\begin{equation}\label{eq:nu_j}
    \sum_{x \in S_j} w_j(x) \|x-\nu^j\|^p.
\end{equation}
Indeed if $\nu^j$ does not minimize this quantity, we can change it and reduce the cost of \eqref{eq:bary def}.

Consider the case of $p=2$. For a fixed $j$, \eqref{eq:nu_j} is just a weighted $k$-means problem whose solution is the weighted average of the points in $S_j$. To prove this, consider taking the gradient of \eqref{eq:nu_j} with respect to the $k$-th coordinate of $\nu^j$. Then setting it equal to $0$ gives us that the $k$-th coordinate will be the weighted average of the $k$-th coordinates of the points  $S_j$. That is, we have
\begin{equation}\label{eq:opt center}
    \v^j = \frac{\sum_{x \in S_j} w_j(x) x}{\sum_{x \in S_j} w_j(x)} = \frac{1}{kb_j}\sum_{x \in S_j} w_j(x) x.
\end{equation}
The second equality results from observing that in order for the $w_j$s to define a proper coupling, we have $\sum_{j=1}^n w_j(x) = a(x)$ for all $x$ in the support of the $\mu_i$s, and $\sum_{x \in \mathrm{supp}(\mu_i)} w_j(x) = b_j$ for all $i$, along with $w_j(x) \geq 0$. In particular, this implies that $\sum_{x \in S_j} w_j(x) = kb_j$ for all $j=1,\ldots,n$.

For arbitrary $p$, such a concise description of $\nu^j$ is not possible. Therefore an alternate, but equivalent, way to characterize the distribution $\nu$ is to just define the sets $S_j$ and weight functions $w_j(\cdot)$ for $1 \le j \le n$. This motivates the following definitions.

\begin{definition}\label{def:solution}
A solution $(S, w) = (S_1, \ldots, S_n, w_1, \ldots, w_j)$ is a valid partition as described previously (meaning that these partitions come from the optimal coupling between each $\mu_i$ to a \emph{fixed} $\nu$), along with the corresponding weight functions $w_j(\cdot)$.
\end{definition}

\begin{figure}[!ht]
\centering
\begin{tikzpicture}
\draw[fill=red] (-0.3,1.6) circle (0.05);
\draw[fill=blue] (-0.3,1) circle (0.05);
\draw[fill=green] (0.6,1.6) circle (0.05);
\node at (0, 1.4){\small{x}};
\draw (0, 1.4) circle (0.8);

\draw[fill=red] (1+1,1+0) circle (0.05);
\draw[fill=blue] (1+1.4,1+0.2) circle (0.05);
\draw[fill=green] (1+1.5,1+-0.5) circle (0.05);
\node at (1+1.3, 1+-0.1){\small{x}};
\draw (1+1.3, 1+-0.1) circle (0.8);

\draw[fill=red] (2+2,2) circle (0.05);
\draw[fill=blue] (2+2.3,1.5) circle (0.05);
\draw[fill=green] (2+2.6,2.2) circle (0.05);
\node at (2+2.3, 1.9){\small{x}};
\draw (2+2.3, 1.9) circle (0.8);
\end{tikzpicture}
\caption{Points of the same color belong to the same distribution. The sets $S_j$ are denoted by the large black circles. Given the partitions $S_j$ (denoted by large black circles) and associated weight functions $w_j$, we can reconstruct the barycenter (denoted by crosses).}
\label{fig:example}
\centering
\end{figure}
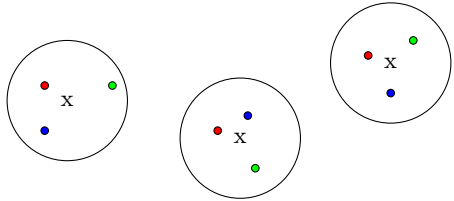

\begin{definition} \label{def:cost} Let $(S, w)$ be a solution. The cost of this solution, denoted $\C_p(S)$, is the value of the objective \eqref{eq:bary def} when we reconstruct $\nu$ from $S$ and $w$ and evaluate \eqref{eq:bary def}:
$$ \C_p(S) = \min_{\nu} \frac{1}{k} \sum_{j=1}^n \sum_{x\in S_j} w_j(x) \| x - \nu^j \|^p. $$
Similarly for a projection $\pi$, $\C_p(\pi S)$ denotes the value of the objective \eqref{eq:bary def} when we first project each of the distributions to $\R^m$ using $\pi$, then compute $\widetilde{\nu}$ using the original weights $w_j$:
$$ \C_p(\pi S) = \min_{\widetilde{\nu}} \frac{1}{k} \sum_{j=1}^n \sum_{x\in S_j} w_j(x) \| \pi(x) - \widetilde{\nu}^j \|^p. $$
Note that each $\widetilde{\nu}^j \in \R^m$.
We suppress the dependence of the cost on $w$ for notational convenience.
\end{definition}

For the case of $p=2$, we can further massage the value of $\nu^j$ in \eqref{eq:nu_j}. Let $\bar{x}$ denote the weighted average of points in $S_j$ (given by \eqref{eq:opt center}). From our discussion above, we know that $\nu^j = \bar{x}$. After some standard algebraic manipulation, we can show that
$ \sum_{x \in S_j} w_j(x) \|x-\nu^j\|^2 = \sum_{x \in S_j} w_j(x) \|x\|^2  - kb_j\|\bar{x} \|^2$ and $
    \sum_{x, y\in S_j} w_j(x)w_j(y) \|x-y\|^2 = 2kb_j\left( \sum_{x \in S_j} w_j(x) \|x\|^2  - kb_j\|\bar{x} \|^2\right)$. 
Combining these equations yields the following for the $p=2$ objective. 
\begin{equation}\label{eq:equiv_obj}
    \frac{1}{2kb_j} \left(  \sum_{x, y\in S_j} w_j(x)w_j(y) \|x-y\|^2  \right) = \sum_{x \in S_j} w_j(x) \|x-\nu^j\|^2.
\end{equation}

\textbf{Dimension reduction.}
In this paper we are concerned with dimensionality reduction maps $\pi: \R^d \rightarrow \R^m$ that are JL projections, i.e., any dimensionality reduction map that satisfies the condition of the JL lemma. This includes random Gaussian and sub-Gaussian matrices \cite{LarsenN16, MakarychevMR19}. We are mainly concerned with making the projection dimension $m$ as small as possible.  

Consider any algorithm $\mathcal{A}$ that, given $\mu_1, \cdots, \mu_k$, solves for some approximate or exact $\nu$ minimizing the objective \eqref{eq:bary def}. We can combine any such $\mathcal{A}$ with dimensionality reduction by first projecting the point masses of the $\mu_i$ down to $\R^m$ for some $m < d$ and using $\mathcal{A}$ to compute some barycenter $\widetilde{\nu}$ in $\R^m$. Then, we can consider the solution $(S, w)$ induced by $\widetilde{\nu}$ (see Definitions \ref{def:solution} and \ref{def:cost}) to \emph{reconstruct} the appropriate $\nu$ in the original dimension $\R^d$ using the objective Eq.\ \eqref{eq:nu_j}. Note that this objective is a convex program for any $p \ge 1$ since we are given $S_j$ and $w_j(\cdot)$. For $p=2$ (which is the most common case), $\nu^j$ has a particularly simple form which is the weighted average of the points in $S_j$ (see Eq.\ \eqref{eq:opt center}). This procedure is outlined in Algorithm \ref{alg:DR}.

\begin{algorithm}[!htb]
\caption{Using dimensionality reduction with any algorithm $\mathcal{A}$ for computing WB} 
\label{alg:DR}
\begin{algorithmic}[1]
\REQUIRE{$k$ discrete distributions $\mu_1, \cdots, \mu_k$ with point masses in dimension $\R^d$, projection dimension $m$, algorithm $\mathcal{A}$}
\STATE{Project the point masses of each distribution $\mu_i$ to dimension $\R^m$ using a JL projection}
\STATE{Use algorithm $\mathcal{A}$ to solve (or approximately solve) the Wasserstein barycenter problem in $\R^m$ to get a distribution $\widetilde{\nu}$ \hspace{.1in} \texttt{//$\widetilde{\nu}$ is a discrete distribution in $\R^m$}}
\STATE{Let $(S, w)$ be the solution that partitions the the point masses of the distributions as described in Definition \ref{def:solution}}
\FOR{each $S_j \in S$}
\STATE{Solve for $\nu^j$ minimizing  $\sum_{x \in S_j} w_j(x) \|x-\nu^j\|^p$
\hspace{.1in}\texttt{//This is a convex program for $p\geq 1$. For $p=2$, $\nu^j$ is just the weighted average of points in $S_j$.}}
\ENDFOR
\STATE{Output the distribution $\nu$ supported on $\nu^j$, and where $\nu^j$ has the same weight as $\widetilde{\nu}^j$ }
\end{algorithmic}
\end{algorithm}

As a corollary of our results, if algorithm $\mathcal{A}$ takes time $T(n,k,d)$, then using dimensionality reduction as in the procedure outlined above takes time $T(n,k,m)$ plus the time to perform the projection and reconstruct the barycenter using the solution $S$. The cost of running algorithm $\mathcal{A}$ is usually much more expensive than performing the projection, and the reconstruction step can also be solved efficiently since it is convex. In the case of $p=2$, the reconstruction just amounts to computing $n$ weighted means. Therefore for $m \ll d$, we get significant savings since $T(n,k,m) \ll T(n,k,d)$.

\section{Reduction to $O(\log(nk))$ Dimensions}\label{sec:log(nk) reduction}

We first show that it suffices to project the point masses of the input distribution into $O(\log(nk))$ dimensions and guarantee that the cost of any \emph{solution} is preserved. Note that our results hold \emph{simultaneously for all solutions.} We first state the $p=2$ case.

\begin{theorem}\label{thm:p=2}
Consider a JL projection $\pi$ from $\mathbb{R}^d$ to $\mathbb{R}^m$ for $m = O(\log(nk/\delta)/\varepsilon^2)$. Then 
\[\mathbb{P}\left( \C_2(\pi S) \in [ (1-\varepsilon)^2 \cdot \C_2(S) \, , (1+\varepsilon)^2 \cdot \C_2(S)] \text{ for all solutions } S \right) \ge 1-\delta. \]
\end{theorem}
\begin{proof}
The proof follows from the solution decomposition given in \eqref{eq:equiv_obj} if we condition on all the pairwise distances being preserved which happens with probability $1-\delta$.
\end{proof}

A decomposition similar to \eqref{eq:equiv_obj} does not exist for $p\ne2$. To prove an analogous theorem for $p \ne 2$, we need the following Lipschitz extension theorem which roughly allows us to ``invert'' a dimensionality reduction map.

\begin{theorem}[Kirszbraun Theorem \cite{Kirszbraun1934}]\label{thm:extension} For any $D \subset \mathbb{R}^m$, let $f : D \rightarrow \mathbb{R}^d$ be an $L$-Lipschitz function.
Then there exists some extension $\widetilde{f}: \mathbb{R}^m \rightarrow \mathbb{R}^d$ of $f$ to the entirety of $\mathbb{R}^m$ such that $f(x) = \widetilde{f}(x)$ for all $x \in D$ and $\widetilde{f}$ is also $L$-Lipschitz.
\end{theorem}

The Kirszbraun theorem allows us to prove Theorem \ref{thm:p=2} for general $p$ with a dimension bound of $m = O(\log(nk/\delta)p^2/\varepsilon^2)$ (see Theorem \ref{thm:allp} in Supplementary Section \ref{sec:lognk_proofs}).

The overview for the proof strategy for the general $p\ne2$ case is as follows. First suppose that all the pairwise distances between the support points of all the distributions are preserved under the projection map up to multiplicative error $1\pm \eps$. This event happens with probability at least $1-\delta$. We then consider the map $f : \R^m \rightarrow \R^d$ that maps each of the \emph{projected} points to its original counterpart in $\R^d$. Note that the map is from the \emph{smaller} dimension $m$ to the larger dimension $d$. On the support points, we know that $f$ is $(1+\eps)$-Lipschitz by our assumption above. 
 
Now if the projection caused the cost of $\pi S$ to decrease significantly, then using the Kriszbraun theorem, one could ``lift'' the corresponding barycenter $\widetilde{\nu}$ from the projected dimension to the original dimension using the extension map $\widetilde{f}$. Then since $\widetilde{f}$ is Lipschitz, this lifted barycenter $\widetilde{f}(\widetilde{\nu})$ plugged into Eq. \eqref{eq:nu_j} would subsequently have cost smaller than the original barycenter that corresponds $S$ in the original dimension. This is a contradiction in light of Eq.\ \eqref{eq:nu_j} and the description of $\nu$ given in Section \ref{sec:prelims}. Note that the exact description of $\widetilde{f}$ does not matter for the analysis, just that such a map exists. A complete, rigorous proof can be found in the supplementary section.

\section{Optimal Dimensionality Reduction} \label{sec:optimal dim}
We now present our main theorem which improves the guarantees of Theorems \ref{thm:p=2} and \ref{thm:allp}.

\begin{restatable}{theorem}{thmjlmain}
\label{thm:jl:main}
Let $\mu_1,\ldots,\mu_k$ be discrete probability distributions on $\R^d$ such that $|\textrm{supp}(\mu_i)| \le \poly(n)$ for all $i$.
Let $d\ge 1$, $\eps,\delta\in(0,1)$, and $p\ge1$. 
Let $\pi_{d,m}:\mathbb{R}^d\to\mathbb{R}^m$ be a family of random JL maps with $m=O\left(\frac{p^4}{\eps^2}\log\frac{n}{\eps\delta}\right)$. 
Then we have,
\[\mathbb{P}\left( \C_p(\pi S) \in [ (1-\varepsilon) \cdot \C_p(S) \, , (1+\varepsilon) \cdot \C_p(S)] \text{ for all solutions } S\right) \ge 1-\delta. \]
\end{restatable}

We now give an overview of the proof strategy for Theorem \ref{thm:jl:main}, deferring all technical details to the supplementary section. 
Ideally, one would like to use a strategy similar to the proof of Theorem \ref{thm:allp}. The key bottleneck is that when we project down to the $m$ specified in Theorem \ref{thm:jl:main}, a large number of pairwise distances between the support points of the $k$ distributions can be distorted (since we are projecting to a dimension smaller than $O(\log(nk))$). Therefore, the Kirszbraun theorem cannot apply as the map $f$ described in the proof strategy of Theorem \ref{thm:allp} is no longer Lipschitz on the support points. 

To overcome this barrier, we generalize an approach of \cite{MakarychevMR19}, who achieved the optimal dimensionality bounds for $k$-means clustering beyond the na\"{i}ve JL bound by defining a \emph{distortion graph} on the set of input points, which has an edge between each pair of points if their pairwise distance is distorted by at least a $(1+\eps)$-factor under the random projection map $\pi$. 
They show that the distortion graph is everywhere sparse, i.e., each vertex has small expected degree in the distortion graph, which implies a ``robust'' Kirszbraun theorem (for their particular problem of $k$-means clustering). 
Namely, there exists an extension map $\widetilde{f}:\mathbb{R}^d\to\mathbb{R}^m$ and a specific point $v\in\mathbb{R}^m$ in the projected space such that a large fraction of the distances from the pre-image $\widetilde{f}^{-1}(v)$ to the input points in $\mathbb{R}^d$ are preserved. 
Moreover, the input points whose distance to $\widetilde{f}^{-1}(v)$ is not preserved can be shown to contribute small error to the $k$-means clustering cost. 

The dimensionality reduction maps of Theorem~\ref{thm:jl:main} generally require multiplication by a dense matrix of (scaled) subgaussian variables. In the Supplementary Section, we show that ``faster'' dimensionality reduction maps can also be used by providing a trade off between the projection runtime and the dimension $m$. Note that in practice, performing the projection is extremely cheap since we only need to perform one matrix multiplication, which is highly optimized. Therefore the cost of any algorithm for Wasserstein barycenter will typically outweigh the cost of computing the projection.

\subsection{Dimensionality Reduction Lower Bounds} \label{sec:lower bounds}
In this section, we state lower bounds on the projection dimension $m$ for the Wasserstein barycenter problem. Theorem \ref{thm:wb:lb} shows that Theorem \ref{thm:jl:main} is tight up to constant factors.

\begin{restatable}{theorem}{thmwblb}
\label{thm:wb:lb}
Consider the setup of Theorem \ref{thm:jl:main}. Any Gaussian matrix used as a dimension reduction map that allows a $(1+\eps)$-approximation to the optimal Wasserstein barycenter requires dimension $\Omega(\log n/\eps^2)$.
\end{restatable}




We also prove that one cannot do better than the na\"{i}ve JL bound for the related problem of computing the optimal transport between two discrete distributions with $n$ point masses each. This is in contrast to the case of Wasserstein barycenter where we were able to overcome the bound that comes from the JL lemma alone. 
Theorem \ref{thm:bad pullback} shows that the optimal solution in the projected dimension can induce a poor quality solution in the original dimension if the projection dimension is smaller than $\log n$.

\begin{restatable}{theorem}{thmbadpullback}
\label{thm:bad pullback}
There exists point sets $A, B \subset \mathbb{R}^d$ with $|A| = |B| = n$ and matching cost $M$ between them, such that if randomly projected down to $m = o(\log n)$ dimensions using an appropriately scaled Gaussian random matrix, the pull back cost of the optimal matching in $\mathbb{R}^m$ is at least $\omega(M)$.
\end{restatable}

In addition, we prove a related theorem which states that the \emph{cost} of the optimal transport is heavily distorted if we project to fewer than $\log n$ dimensions.

\begin{theorem} \label{thm:bad reduction}
There exists point sets $A, B \subset \mathbb{R}^d$ with $|A| = |B| = n$ and matching cost $M$ between them, such that if randomly projected down to $m = o(\log n)$ dimensions using an appropriately scaled Gaussian random matrix, the cost of optimal matching in $\mathbb{R}^m$ is $o(M)$ with probability at least $2/3$.
\end{theorem}

See Supplementary Section \ref{sec:supp_lower_bounds} for full proofs.

\section{Coresets} \label{sec:coresets}
In this section, we give a coreset construction for Wasserstein barycenters. Our goal is to reduce the number of distributions $k$ to only depend polynomially on $n$. We first define our notion of coresets.

\begin{definition}[Coreset]\label{def:coreset}
Fix $p \ge 1$. Let $C$ and $M$ be two sets of distributions in $\R^d$ where all distributions consist of $\poly(n)$ point masses. $C$ is called an \emph{$\varepsilon$-corset} for the set of distributions $M$ if there exist weights $w_c$ for $ c \in C$ such that for all distributions $\nu$ of support size at most $n$, it holds that
\[ (1-\varepsilon) \, \sum_{c \in C}w_c \, W(c, \nu)^p \le \frac{1}{|M|}\sum_{\mu \in M}W(\mu, \nu)^p \, \le (1+\varepsilon) \sum_{c \in C}w_c \, W(c, \nu)^p .\]
\end{definition}
The main result of this section is the following theorem.
\begin{theorem}[Theorem \ref{thm:final_coreset_bound} simplified]\label{thm:coreset_informal} Let $M$ be a set of discrete distributions in $\R^d$, each supported on at most $\poly(n)$ point masses. There exists a weighted subset $K \subseteq M$ of size $\poly(n,d)/\eps^2$ that satisfies Definition \ref{def:coreset} for $p=O(1)$.
\end{theorem}

To prove Theorem \ref{thm:coreset_informal}, we follow the ``importance sampling'' by sensitivities framework in conjunction with using structural properties of the Wasserstein barycenter problem itself. The sensitivity sampling framework has been successfully applied to achieve corsets for many problems in machine learning (see the references in the survey \cite{bachem2017practical}). Note that we have not attempted to optimize the constants in our proofs and instead focus on showing that $k$ can be reduced to $\poly(n,d)$ for simplicity. The formal proof of Theorem \ref{thm:coreset_informal} is deferred to the supplementary section.

We now describe the high level overview of the proof. We form the set $C$ by sampling distributions in $M$ with replacement based on their ``importance'' or contribution to the total cost. The notion of importance is formally captured by the definition of sensitivity.

\begin{definition}[Sensitivity]\label{def:sensitivity}
Consider the set $N$ of all possible barycenter distributions $\nu$ with support size at most $n$.  
The sensitivity of a distribution $\mu \in M$ is defined as 
\[\sigma(\mu) = \sup_{\nu \in N} \, \frac{W(\mu,\nu)^p}{ \frac{1}{|M|}\sum_{\mu' \in M} W(\mu', \nu)^p }. \]
The total sensitivity is defined as $\mathfrak{S} = \frac{1}{|M|}\sum_{\mu \in M} \sigma(\mu)$.
\end{definition}

 To see why such a notion is beneficial, consider the case that one distribution $\mu$ consists of point masses that are outliers among all of the point masses comprising the distributions in $M$. Then it is clear that we must sample $\mu$ with a higher probability if we wish to satisfy the definition of a coreset. In particular, we sample each distribution in $M$ with probability proportional to (an upper bound on) its sensitivity. Using a standard result in coreset construction, we can bound the size of the coreset in terms of the total sensitivity and a measure of the ``complexity'' of the Wasserstein barycenter problem which is related to the VC dimension. In particular, we utilize the notion of psuedo-dimension.
 
 \begin{definition}[Pseudo-Dimension, Definition $9$ \cite{Feldman_gaussians}]
Let $\mathcal{X}$ be  a  ground  set  and $\mathcal{F}$ be a set of functions from $\mathcal{X}$ to the interval $[0,1]$. Fix a set $S= \{x_1,\cdots ,x_n\} \subset \mathcal{X}$, a set of reals numbers $R=\{r_1, \cdots, r_n\}$ with $r_i \in [0, 1]$ and a function $f \in \mathcal{F}$. The set $S_f= \{x_i \in S  \mid f(x_i) \ge r_i\}$ is called the induced subset of $S$ formed by $f$ and $R$. The set $S$ with associated values $R$ is shattered by $\mathcal{F}$ if $|\{S_f \mid f \in \mathcal{F} \}|= 2^n$. 
The \emph{pseudo-dimension} of $\mathcal{F}$ is the cardinality of the largest shattered subset of $\mathcal{X}$ (or $\infty$).
\end{definition}
 
 The following theorem provides a formal connection between the size of coresets and the notion of sensitivity and psuedo-dimension. Note that the statement of the theorem is more general and applies to a wider class of problems. However, we specialize the theorem statement to the case of Wasserstein Barycenters.

\begin{theorem}[Coreset Size, Theorem $2.4.6$ in \cite{lang_thesis}, Theorem $2.3$ in \cite{bachem2017practical} for the case of Wasserstein Barycenters]\label{thm:coreset_size} Let $\varepsilon > 0$ and $\delta \in (0,1)$. Let $s: M \rightarrow \R^{\ge 0}$ denote any upper bound function on the sensitivity $\sigma(\cdot)$ defined in Definition \ref{def:sensitivity} and let $S = \frac{1}{|M|} \sum_{\mu \in M} s(\mu)$. Consider a set $K$ of $|K|$ samples of $M$ with replacement where each distribution $\mu \in M$ is sampled with probability $q(\mu) = s(\mu)/(|M| \cdot S)$ and each sampled point is assigned the weight $1/(|M| \cdot |K|\cdot q(\mu))$. Let $\mathcal{F}$ denote the set of functions 
\[ \mathcal{F} = \left\{  \frac{W(\cdot, \nu)^p}{ Sq(\cdot) \, \sum_{\mu \in M}  W(\mu,\nu)^p} \mid \nu \in N \right\} \]
where $N$ is the set of all possible barycenter distributions with support size at most $n$. 
Let $d'$ denote the pseudo-dimension of $\mathcal{F}$. Then the set $K$ (along with the associated weights) satisfies Definition \ref{def:coreset} with probability at least $1-\delta$ if 
\[|K| \ge \frac{cS}{\varepsilon^2}\left(d'\log S + \log \frac{1}{\delta} \right)  \]
where $c> 0$ is some absolute constant. 
\end{theorem}
 
Thus, the bulk of our work lies in bounding the sensitivities and psuedo-dimension. For the former quantity, we exploit the fact that the Wasserstein distance is a metric. The latter requires us to use tools from statistical learning theory which relate the VC dimension of a function class to its algorithmic complexity (see Lemmas \ref{lem:PD_to_VC} and \ref{lem:VC_bound}). Full details given in Supplementary section \ref{sec:coreset_proofs}.

\section{Other Theoretical Results}\label{sec:others}
We now present some additional theoretical results pertaining to Wasserstein barycenters. Our first result is that Wasserstein barycenters can be formulated as a constrained low-rank approximation problem. This class of problems includes coputing the SVD and $k$-means clustering \cite{CohenEMMP15}. Formally, we prove the following theorem.

\begin{theorem}\label{thm:low_rank}
Given discrete distributions $\mu_1, \ldots, \mu_k \in \R^d$ with support size at most $n$, consider the problem of computing the Wasserstein barycenter with support size at most $n$ for the $p=2$ objective. There exists a matrix $A \in \R^{nk \times d}$ and a set $S$ of rank $n$ orthogonal projection matrices in $\R^{nk \times nk}$ such that the first problem is equivalent to computing 
\[\bP^*=\argmin_{\bP\in S}\|\bA-\bP\bA\|_F^2.\]
\end{theorem}

The proof of Theorem \ref{thm:low_rank} is given in Section \ref{sec:sup_constrained}.

We also prove the following NP hardness result in Section \ref{sec:np_hard} which complements the hardness results in \cite{AltschulerB21, BorgwardtP21}.

\begin{theorem}
\label{thm:wb:apx}
It is NP-hard to approximate an optimal Wasserstein barycenter of fixed support size up to a multiplicative factor $1.0013$.
\end{theorem}

\section{Experiments}\label{sec:experiments}
In this section, we empirically verify that dimensionality reduction can provide large computational savings without significantly reducing accuracy. We use the following datasets in our experiments.

\noindent \textbf{FACES dataset}: This dataset is used in the influential ISOMAP paper and consists of $698$ images of faces in dimension $4096$ \cite{isomap}. We form $k=2$ distributions by splitting the images facing to the ``left'' versus the ones facing ``right.'' This results in $\sim 350$ uniform point masses per distribution.

\noindent  \textbf{MNIST dataset}: We subsample $10^4$ images from the MNIST test dataset (dimension $784$). We split the images by their digit class which results in $k=10$ distributions with $\sim 10^3$ uniform point masses each in $\R^{784}$.

\textbf{Experimental setup.}
We project our datasets in dimensions $d$ ranging from $d=2$ to $d=30$ and compute the Wasserstein barycenter for $p=2$. For FACES, we limit the support size of the barycenter to be at most $5$ points in $\R^{4096}$ (since the barycenter should intuitively return an ``interpolation'' between the left and right facing faces, it should not be supported on too many points). For MNIST we limit the support size of the barycenter to be at most $40$. We then take the barycenter found in the lower dimension and compare its cost in the higher dimension (see Algorithm \ref{alg:DR}) against the Wasserstein barycenter found in the higher dimension.

We use the code and default settings from \cite{code} to compute the Wasserstein barycenter; this implementation has been applied in previous empirical papers \cite{Ye2017FastDD}. 
While we fix this implementation, note that dimensionality reduction is extremely flexible and can work with any algorithm or implementation 
(see Algorithm \ref{alg:DR}) and we would expect it to produce similar results.

\textbf{Results.}
Our results are displayed in Figure \ref{fig:experiments}. We see that for both datasets, reducing the dimension to $d=30$ only increases the cost of the solution by $5\%$. This is \textbf{1-2} orders of magnitude smaller than from the original dimensions of $784$ and $4096$ for MNIST and FACES respectively. The average time taken to run the Wasserstein barycenter computation algorithm in $d=30$ was $73\%$ and $9\%$ of the time taken to run in the full dimensions respectively.

\begin{figure}[!htb]
\centering
\begin{subfigure}{.5\textwidth}
  \centering
  \includegraphics[width=.7\linewidth]{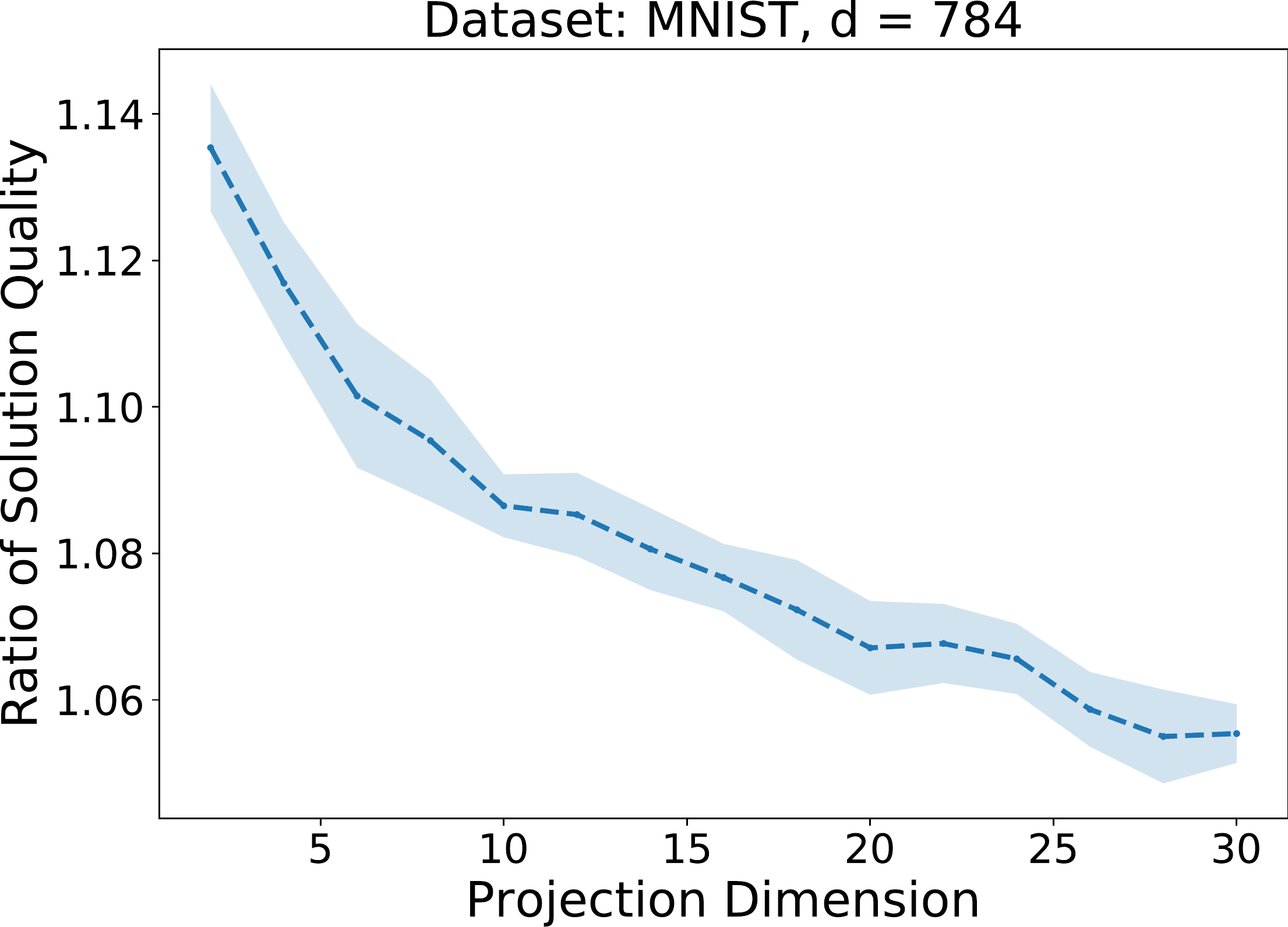}
  \caption{MNIST}
  \label{fig:mnist}
\end{subfigure}%
\begin{subfigure}{.5\textwidth}
  \centering
  \includegraphics[width=.7\linewidth]{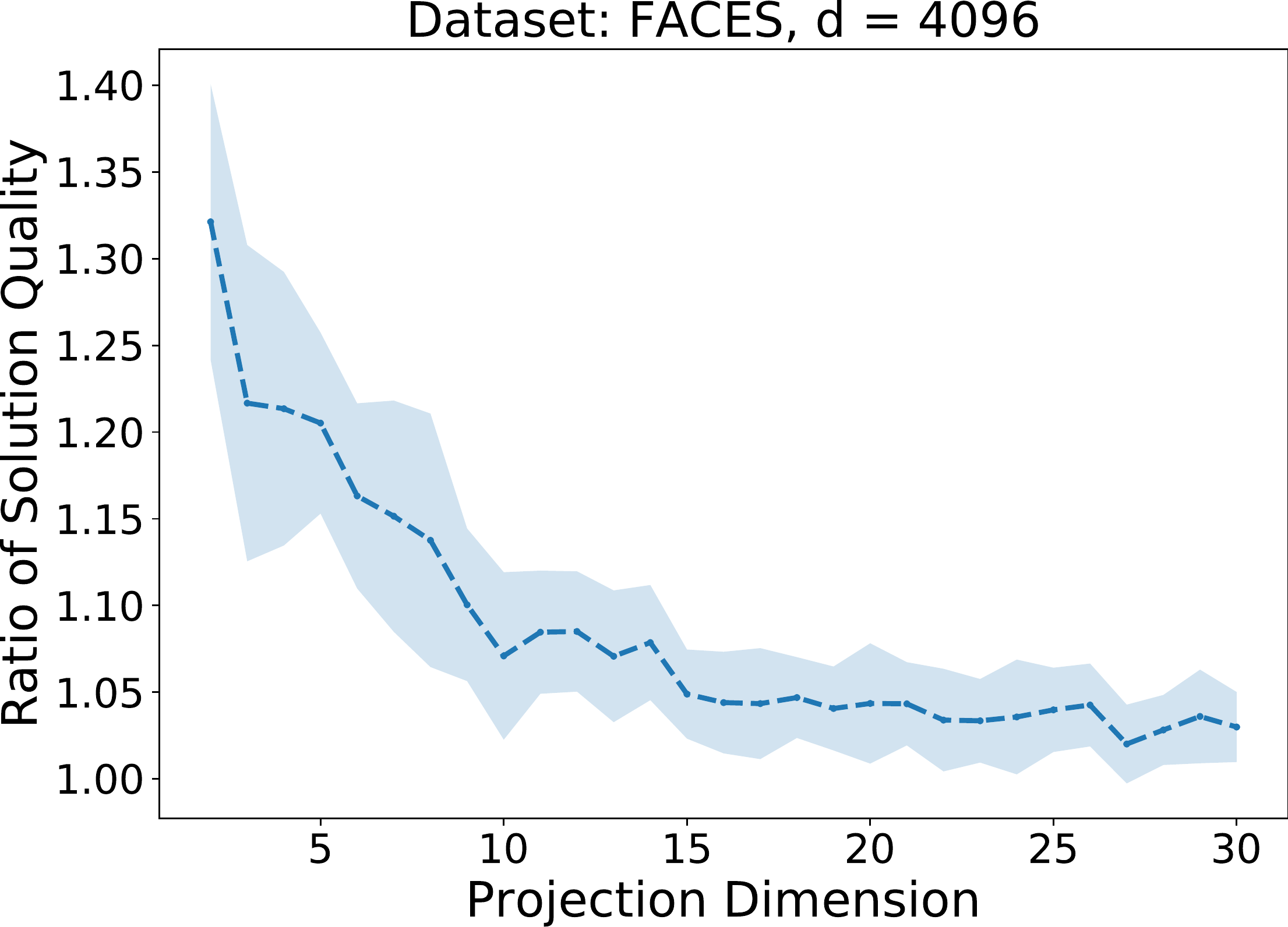}
  \caption{Faces}
  \label{fig:faces}
\end{subfigure}
\caption{Ratio of the quality of solution found in the lower dimension versus the original dimension. Result displays average of $20$ independent trials and $\pm 1$ standard deviation is shaded.}
\label{fig:experiments}
\end{figure}

\paragraph{Coreset experiments.}
Our coreset result reduces the number of distributions $k$ through sensitivity (importance) sampling. 
We created a synthetic dataset with large $k$ but small $n$ and $d$ to emphasize the advantage of sensitivity sampling over uniform sampling. 
We have $k=50,000$ distributions that each consists of a single point mass in $\R$. 
The first $k-1$ distributions are all supported at the origin while one distribution is supported at $x = k$. 
We consider the $p=2$ case and limit the support size of the barycenter to also be $1$. Let $\C_{\textrm{orig}}(\nu)$ denote the cost of $\nu$ on the original objective \eqref{eq:bary def} and let $\C_{\textrm{core}}(\nu)$ the cost of \eqref{eq:bary def} when evaluated on a coreset.  We record the relative error 
$|\C_{\textrm{core}}(\nu)-\C_{\textrm{orig}}(\nu)| / |\C_{\textrm{orig}}(\nu)|$
evaluated at $\nu = \delta_x$, i.e. a single unit point mass at $x$, for $x=0, \, 1, \, 10$. We then average the results across $10$ trials each. As $x$ (the point on which the query distribution is supported) grows bigger, the associated cost became bigger, hence decreasing the relative error. 
Other query locations displayed the same trend. 
See Figure~\ref{fig:coreset} for more details.

\begin{figure}[!htb]
\centering
\begin{tabular}{|c|c|c|c|c|c|}\hline
\multirow{2}{*}{Method} & \multirow{2}{*}{\# of samples} & \multicolumn{4}{|c|}{\% error at query}\\\cline{3-6}
& & $x=100$ & $x=10$ & $x=1$ & $x=0$ \\\hline
Uniform sampling & \textcolor{purple}{1000} & $0.986$ & $9.087$ & $49.998$ & $100$ \\\hline
Sensitivity sampling & \textcolor{blue}{10} & $0.0040$ & $0.0036$ & $0.0020$ & $0$\\\hline
\end{tabular}
\caption{Even with much fewer samples, sensitivity sampling outperforms uniform sampling for a number of query locations, averaged across $10$ repetitions.}\label{fig:coreset}
\end{figure}


\section*{Acknowledgments}
Sandeep Silwal was supported in part by a NSF Graduate Research Fellowship Program. 
Samson Zhou was supported by a Simons Investigator Award of David P. Woodruff.

\bibliography{main}

\newcommand{\etalchar}[1]{$^{#1}$}
\begin{thebibliography}{SDGP{\etalchar{+}}15}

\bibitem[AB99]{anthony_bartlett_1999}
Martin Anthony and Peter~L. Bartlett.
\newblock {\em Neural Network Learning: Theoretical Foundations}.
\newblock Cambridge University Press, 1999.

\bibitem[AB21]{AltschulerB21}
Jason~M. Altschuler and Enric Boix{-}Adser{\`{a}}.
\newblock Wasserstein barycenters are np-hard to compute.
\newblock {\em CoRR}, abs/2101.01100, 2021.

\bibitem[ABA21]{altschuler2021wasserstein}
Jason~M Altschuler and Enric Boix-Adsera.
\newblock Wasserstein barycenters can be computed in polynomial time in fixed
  dimension.
\newblock {\em Journal of Machine Learning Research}, 22(44):1--19, 2021.

\bibitem[ABM16]{anderes16}
Ethan Anderes, Steffen Borgwardt, and Jacob Miller.
\newblock Discrete wasserstein barycenters: Optimal transport for discrete
  data.
\newblock {\em Mathematical Methods of Operations Research}, 84, 10 2016.

\bibitem[AC09]{AilonC09}
Nir Ailon and Bernard Chazelle.
\newblock The fast johnson--lindenstrauss transform and approximate nearest
  neighbors.
\newblock {\em {SIAM} J. Comput.}, 39(1):302--322, 2009.

\bibitem[AC11a]{martial11}
Martial Agueh and Guillaume Carlier.
\newblock Barycenters in the wasserstein space.
\newblock {\em SIAM J. Math. Analysis}, 43:904--924, 01 2011.

\bibitem[AC11b]{agueh2011barycenters}
Martial Agueh and Guillaume Carlier.
\newblock Barycenters in the wasserstein space.
\newblock {\em SIAM Journal on Mathematical Analysis}, 43(2):904--924, 2011.

\bibitem[{\'A}DCM16]{alvarez2016fixed}
Pedro~C {{\'A}lvarez-Esteban}, E~{Del Barrio}, JA~{Cuesta-Albertos}, and
  C~Matr{\'a}n.
\newblock A fixed-point approach to barycenters in wasserstein space.
\newblock {\em Journal of Mathematical Analysis and Applications},
  441(2):744--762, 2016.

\bibitem[BBC{\etalchar{+}}19]{BecchettiBC0S19}
Luca Becchetti, Marc Bury, Vincent Cohen{-}Addad, Fabrizio Grandoni, and Chris
  Schwiegelshohn.
\newblock Oblivious dimension reduction for \emph{k}-means: beyond subspaces
  and the johnson-lindenstrauss lemma.
\newblock In {\em Proceedings of the 51st Annual {ACM} {SIGACT} Symposium on
  Theory of Computing, {STOC}}, pages 1039--1050, 2019.

\bibitem[BCC{\etalchar{+}}15]{benamou2015iterative}
Jean-David Benamou, Guillaume Carlier, Marco Cuturi, Luca Nenna, and Gabriel
  Peyr{\'e}.
\newblock Iterative bregman projections for regularized transportation
  problems.
\newblock {\em SIAM Journal on Scientific Computing}, 37(2):A1111--A1138, 2015.

\bibitem[BLK17]{bachem2017practical}
Olivier Bachem, Mario Lucic, and Andreas Krause.
\newblock Practical coreset constructions for machine learning, 2017.

\bibitem[BP21]{BorgwardtP21}
Steffen Borgwardt and Stephan Patterson.
\newblock On the computational complexity of finding a sparse wasserstein
  barycenter.
\newblock {\em J. Comb. Optim.}, 41(3):736--761, 2021.

\bibitem[BT97]{bertsimas1997introduction}
Dimitris Bertsimas and John~N Tsitsiklis.
\newblock {\em Introduction to linear optimization}, volume~6.
\newblock Athena Scientific Belmont, MA, 1997.

\bibitem[CC06]{ChlebikC06}
Miroslav Chleb{\'{\i}}k and Janka Chleb{\'{\i}}kov{\'{a}}.
\newblock Complexity of approximating bounded variants of optimization
  problems.
\newblock {\em Theor. Comput. Sci.}, 354(3):320--338, 2006.

\bibitem[CD14]{cuturi2014fast}
Marco Cuturi and Arnaud Doucet.
\newblock Fast computation of wasserstein barycenters.
\newblock In {\em International conference on machine learning}, pages
  685--693, 2014.

\bibitem[CEM{\etalchar{+}}15]{CohenEMMP15}
Michael~B. Cohen, Sam Elder, Cameron Musco, Christopher Musco, and Madalina
  Persu.
\newblock Dimensionality reduction for k-means clustering and low rank
  approximation.
\newblock In {\em Proceedings of the Forty-Seventh Annual {ACM} on Symposium on
  Theory of Computing, {STOC}}, pages 163--172, 2015.

\bibitem[CG15]{cunningham2015linear}
John~P Cunningham and Zoubin Ghahramani.
\newblock Linear dimensionality reduction: Survey, insights, and
  generalizations.
\newblock {\em The Journal of Machine Learning Research}, 16(1):2859--2900,
  2015.

\bibitem[CMRS20]{chewi2020gradient}
Sinho Chewi, Tyler Maunu, Philippe Rigollet, and Austin~J Stromme.
\newblock Gradient descent algorithms for bures-wasserstein barycenters.
\newblock In {\em Conference on Learning Theory}, pages 1276--1304, 2020.

\bibitem[COO15]{carlier2015numerical}
Guillaume Carlier, Adam Oberman, and Edouard Oudet.
\newblock Numerical methods for matching for teams and wasserstein barycenters.
\newblock {\em ESAIM: Mathematical Modelling and Numerical Analysis},
  49(6):1621--1642, 2015.

\bibitem[EHJK20]{elvander2020multi}
Filip Elvander, Isabel Haasler, Andreas Jakobsson, and Johan Karlsson.
\newblock Multi-marginal optimal transport using partial information with
  applications in robust localization and sensor fusion.
\newblock {\em Signal Processing}, 171:107474, 2020.

\bibitem[Fel20]{Feldman20}
Dan Feldman.
\newblock Introduction to core-sets: an updated survey.
\newblock {\em CoRR}, abs/2011.09384, 2020.

\bibitem[HNY{\etalchar{+}}17]{ho2017multilevel}
Nhat Ho, XuanLong Nguyen, Mikhail Yurochkin, Hung~Hai Bui, Viet Huynh, and Dinh
  Phung.
\newblock Multilevel clustering via wasserstein means.
\newblock In {\em International Conference on Machine Learning}, pages
  1501--1509. PMLR, 2017.

\bibitem[IN07]{indyknaor}
Piotr Indyk and Assaf Naor.
\newblock Nearest-neighbor-preserving embeddings.
\newblock {\em ACM Trans. Algorithms}, 3(3):31–es, August 2007.

\bibitem[JL84]{johnson1984extensions}
William~B Johnson and Joram Lindenstrauss.
\newblock Extensions of lipschitz mappings into a hilbert space.
\newblock {\em Contemporary mathematics}, 26(189-206):1, 1984.

\bibitem[Kir34]{Kirszbraun1934}
M.~Kirszbraun.
\newblock Über die zusammenziehende und lipschitzsche transformationen.
\newblock {\em Fundamenta Mathematicae}, 22(1):77--108, 1934.

\bibitem[KMN11]{KaneMN11}
Daniel~M. Kane, Raghu Meka, and Jelani Nelson.
\newblock Almost optimal explicit johnson-lindenstrauss families.
\newblock In {\em Approximation, Randomization, and Combinatorial Optimization.
  Algorithms and Techniques - 14th International Workshop, {APPROX}, and 15th
  International Workshop, {RANDOM}. Proceedings}, pages 628--639, 2011.

\bibitem[KTD{\etalchar{+}}19]{kroshnin2019complexity}
Alexey Kroshnin, Nazarii Tupitsa, Darina Dvinskikh, Pavel Dvurechensky,
  Alexander Gasnikov, and Cesar Uribe.
\newblock On the complexity of approximating wasserstein barycenters.
\newblock In {\em International conference on machine learning}, pages
  3530--3540, 2019.

\bibitem[Lan18]{lang_thesis}
Harry Lang.
\newblock {\em Streaming Coresets for High Dimensional Geometry}.
\newblock PhD thesis, Johns Hopkins University, 7 2018.

\bibitem[LFKF18]{Feldman_gaussians}
Mario Lucic, Matthew Faulkner, Andreas Krause, and Dan Feldman.
\newblock Training gaussian mixture models at scale via coresets.
\newblock {\em Journal of Machine Learning Research}, 18(160):1--25, 2018.

\bibitem[LHC{\etalchar{+}}20]{lin2020fixed}
Tianyi Lin, Nhat Ho, Xi~Chen, Marco Cuturi, and Michael~I Jordan.
\newblock Fixed-support wasserstein barycenters: Computational hardness and
  fast algorithm.
\newblock {\em Advances in Neural Information Processing Systems}, 33, 2020.

\bibitem[LN16]{LarsenN16}
Kasper~Green Larsen and Jelani Nelson.
\newblock The johnson-lindenstrauss lemma is optimal for linear dimensionality
  reduction.
\newblock In {\em 43rd International Colloquium on Automata, Languages, and
  Programming, {ICALP}}, 2016.

\bibitem[LN17]{LarsenN17}
Kasper~Green Larsen and Jelani Nelson.
\newblock Optimality of the johnson-lindenstrauss lemma.
\newblock In {\em 58th {IEEE} Annual Symposium on Foundations of Computer
  Science, {FOCS}}, pages 633--638, 2017.

\bibitem[LSPC19]{LuiseSPC19}
Giulia Luise, Saverio Salzo, Massimiliano Pontil, and Carlo Ciliberto.
\newblock Sinkhorn barycenters with free support via frank-wolfe algorithm.
\newblock In {\em Advances in Neural Information Processing Systems 32: Annual
  Conference on Neural Information Processing Systems}, pages 9318--9329, 2019.

\bibitem[LSW17]{LeeSW17}
Euiwoong Lee, Melanie Schmidt, and John Wright.
\newblock Improved and simplified inapproximability for k-means.
\newblock {\em Inf. Process. Lett.}, 120:40--43, 2017.

\bibitem[MC19]{MuzellecC19}
Boris Muzellec and Marco Cuturi.
\newblock Subspace detours: Building transport plans that are optimal on
  subspace projections.
\newblock In {\em Advances in Neural Information Processing Systems 32: Annual
  Conference on Neural Information Processing Systems, NeurIPS}, pages
  6914--6925, 2019.

\bibitem[MMR19]{MakarychevMR19}
Konstantin Makarychev, Yury Makarychev, and Ilya~P. Razenshteyn.
\newblock Performance of johnson-lindenstrauss transform for \emph{k}-means and
  \emph{k}-medians clustering.
\newblock In {\em Proceedings of the 51st Annual {ACM} {SIGACT} Symposium on
  Theory of Computing, {STOC}}, pages 1027--1038, 2019.

\bibitem[PW09]{pele2009fast}
Ofir Pele and Michael Werman.
\newblock Fast and robust earth mover's distances.
\newblock In {\em 2009 IEEE 12th international conference on computer vision},
  pages 460--467. IEEE, 2009.

\bibitem[RGT97]{rubner1997earth}
Yossi Rubner, Leonidas~J Guibas, and Carlo Tomasi.
\newblock The earth mover’s distance, multi-dimensional scaling, and
  color-based image retrieval.
\newblock In {\em Proceedings of the ARPA image understanding workshop}, volume
  661, page 668, 1997.

\bibitem[RPDB11]{rabin2011wasserstein}
Julien Rabin, Gabriel Peyr{\'e}, Julie Delon, and Marc Bernot.
\newblock Wasserstein barycenter and its application to texture mixing.
\newblock In {\em International Conference on Scale Space and Variational
  Methods in Computer Vision}, pages 435--446, 2011.

\bibitem[RU02]{ruschendorf2002n}
Ludger R{\"u}schendorf and Ludger Uckelmann.
\newblock On the n-coupling problem.
\newblock {\em Journal of multivariate analysis}, 81(2):242--258, 2002.

\bibitem[SCSJ17]{staib2017parallel}
Matthew Staib, Sebastian Claici, Justin Solomon, and Stefanie Jegelka.
\newblock Parallel streaming wasserstein barycenters.
\newblock pages 2647--2658, 2017.

\bibitem[SDGP{\etalchar{+}}15]{solomon2015convolutional}
Justin Solomon, Fernando De~Goes, Gabriel Peyr{\'e}, Marco Cuturi, Adrian
  Butscher, Andy Nguyen, Tao Du, and Leonidas Guibas.
\newblock Convolutional wasserstein distances: Efficient optimal transportation
  on geometric domains.
\newblock {\em ACM Transactions on Graphics (TOG)}, 34(4):1--11, 2015.

\bibitem[SLD18]{srivastava2018scalable}
Sanvesh Srivastava, Cheng Li, and David~B Dunson.
\newblock Scalable bayes via barycenter in wasserstein space.
\newblock {\em The Journal of Machine Learning Research}, 19(1):312--346, 2018.

\bibitem[TSL00]{isomap}
Joshua~B. Tenenbaum, Vin~de Silva, and John~C. Langford.
\newblock A global geometric framework for nonlinear dimensionality reduction.
\newblock {\em Science}, 290(5500):2319--2323, 2000.

\bibitem[Vil08]{villani2008optimal}
C{\'e}dric Villani.
\newblock {\em Optimal transport: old and new}, volume 338.
\newblock Springer Science \& Business Media, 2008.

\bibitem[Wai19]{wainwright2019high}
Martin~J Wainwright.
\newblock {\em High-dimensional statistics: A non-asymptotic viewpoint},
  volume~48.
\newblock Cambridge University Press, 2019.

\bibitem[Ye19]{code}
Jianbo Ye.
\newblock Wbc-matlab.
\newblock \url{https://github.com/bobye/WBC_Matlab}, 2019.

\bibitem[YWWL17]{Ye2017FastDD}
Jianbo Ye, P.~Wu, J.~Z. Wang, and Jia Li.
\newblock Fast discrete distribution clustering using wasserstein barycenter
  with sparse support.
\newblock {\em IEEE Transactions on Signal Processing}, 65:2317--2332, 2017.

\end{thebibliography}
\bibliographystyle{alpha}

\newpage
\appendix

\section{Proofs for Section \ref{sec:log(nk) reduction}}\label{sec:lognk_proofs}

\blockcomment{
\begin{proof}[Proof of Theorem \ref{thm:p=2}]
Condition on the event that all the distances between point masses among all of the $nk$ points in the distributions $\mu_1, \cdots, \mu_k$ are preserved up to $1 \pm \varepsilon$ which happens with probability at least $1-\delta$. Consider any solution arbitrary solution $S = (S_1, \ldots, S_n)$. Note that $\C(\pi S)$ decomposes into costs per $S_j$ according to \eqref{eq:nu_j} and each of these costs further decompose into pairwise distances among points in $S_j$ according to \eqref{eq:bary def}. Since we have 
\[ \|\pi x - \pi y\|^2 \in [(1-\varepsilon)^2 \|x-y\|^2, (1+\varepsilon)^2 \|x-y\|^2]  \] for all possible points $x,y$ and all the weights $w_j(x)$ are non negative, it follows that
\begin{align*}
    \C(S_j) &= \sum_{x \in S_j} w_j(x) \|x-\v^j\|^2 \\
    &= \frac{1}{2kb_j} \sum_{x, y \in S_j} w_j(x)w_j(y) \|x - y\|^2 \\
    &\leq \frac{1}{2kb_j} \sum_{x, y \in S_j} w_j(x)w_j(y) (1+\eps)^2\|\pi x - \pi y\|^2 \\
    &= (1+\eps)^2 \C(\pi S_j).
\end{align*}
A similar calculation yields the lower bound $\C(S_j) \geq (1-\eps)^2 \C(\pi S_j)$.
Summing over all $S_j$ and noting that $S$ is arbitrary finishes the proof.
\end{proof}
}

\begin{theorem}\label{thm:allp}
Let $p \ne 2$. Consider a JL projection $\pi$ from $\mathbb{R}^d$ to $\mathbb{R}^m$ for $m = O(\log(nk/\delta)p^2/\varepsilon^2)$. Then we have
\[\mathbb{P}\left( \C_p(\pi S) \in [ (1-\varepsilon) \cdot \C_p(S) \, , (1+\varepsilon) \cdot \C_p(S)]\text{ for all solutions } S\right) \ge 1-\delta, \]
where the probability is taken over the randomness in the projection $\pi$.
\end{theorem}

\begin{proof}[Proof of Theorem \ref{thm:allp}]
We again assume that the distances between point masses among all the $nk$ points in the distributions $\mu_1, \ldots, \mu_k$ are preserved up to $1 \pm \varepsilon$. By Theorem \ref{thm:extension}, dimensionality reduction gives us a $(1+\varepsilon)$-Lipschitz map $\widetilde{\pi} : \mathbb{R}^d \rightarrow \mathbb{R}^m$ as well as $\widetilde{\varphi}: \mathbb{R}^m \rightarrow \mathbb{R}^d$.

Now consider an arbitrary solution $S = (S_1, \ldots, S_n)$. We first show that $\C(S) \le (1+\varepsilon)^p\C(\pi S)$ where $\C(\pi S)$ is the cost of the solution $S$ evaluated in the projected space. Indeed for any $S_j$, the objective in the original dimension $\mathbb{R}^d$ is 
\[\sum_{x \in S_j} w_j(x) \|x-\nu^j\|^p. \]
Let $\nu^j$ denote the argmin of this objective in $\mathbb{R}^d$ and let $u^j$ denote the argmin for this same objective but in the projected space $\mathbb{R}^m$, i.e.
$$u^j = \argmin_{u\in \R^m} \sum_{x\in S_j} w_j(x)\|\pi x-u\|^p.$$
Then we have
\begin{align*}
     \sum_{x \in S_j} w_j(x) \|x-\nu^j\|^p & \le \sum_{x \in S_j} w_j(x) \|x-\widetilde{\varphi}(u^j)\|^p  \qquad \qquad &(\nu^j\text{ is more optimal than } \widetilde{\varphi}(u)) \\
     &= \sum_{x \in S_j} w_j(x)\|\phit(\pi x) - \phit(u^j)\|^p &(\phit = \pi^{-1} \text{ for } x \in S_j) \\
     &\le \sum_{x \in S_j} w_j(x)(1+\varepsilon)^p \|\pi x - u^j \|^p \qquad &(\widetilde{\varphi} \text{ is } (1+\varepsilon)\text{-Lipschitz})  \\
     &= (1+\varepsilon)^p\sum_{x \in S_j} w_j(x) \|\pi x - u^j \|^p.
\end{align*}
Summing over all $S_j$ finally leads to $\C(S) \le (1+\varepsilon)^p\C(\pi S)$. A similar reasoning also gives $\C(\pi S) \le (1+\varepsilon)^p\C(S)$ and combining these two statements and adjusting $\varepsilon$ proves the theorem.
\end{proof}

\section{Proofs for Section \ref{sec:optimal dim}}
In this section, we give the missing proofs from Section~\ref{sec:optimal dim}. Our main goal will be to prove Theorem~\ref{thm:jl:main}; we also describe a ``faster'' dimension reduction map at the end of the section. 
To prove Theorem~\ref{thm:jl:main}, we will actually first prove a version of the theorem with a slightly rescaled value of $\eps$ (Theorem~\ref{thm:wb:main}). Theorem~\ref{thm:jl:main} follows immediately by ``undoing" the rescaling.

\blockcomment{
\begin{restatable}{theorem}{thmwbmain}
\label{thm:wb:main}
Let $\mu_1,\ldots,\mu_k$ be $k$ discrete distributions with support size $n$ on $\R^d$, and let $X = \bigcup_{i=1}^k \mathrm{supp}(\mu_i)$.
Given $\eps\in(0,1/4)$ and $\delta\in(0,1)$ and $\pi:\R^d\to\R^m$ with 
\[m=O\left(\frac{\log\frac{n}{\delta}+p\log\frac{1}{\eps}+p^2}{\eps^2}\right),\]
then with probability at least $1-\delta$, we have that simultaneously for every solution $\calC = (C_1, \ldots, C_n)$ and corresponding weight functions $w_j(\cdot)$ of $X$,
\begin{align*}
\C_p(\pi(\calC))&\le(1+\eps)^{3p}\C_p(\calC)\\
(1-\eps)\C_p(\calC)&\le(1+\eps)^{3p-1}\C_p(\pi(\calC)).
\end{align*}
\end{restatable}
}

We adapt this analysis to the Wasserstein barycenter problem by handling four additional issues: (i) the input points are weighted since they come from probability distributions; (ii) input points may be assigned to multiple support points in the barycenter; (iii) each barycenter point is constrained to receive a specific amount of mass under optimal transport; and (iv) the distorted points must not contribute large error to the cost induced by the Wasserstein barycenter. 
Issues (i) and (ii) are problematic because previous structural results for the distortion graph do not rule out a large weighted fraction of the distances being distorted. 
Furthermore, issues (iii) and (iv) are problematic because we cannot isolate each point in a probability distribution to a specific barycenter. 
We again consider a hypothetical distortion graph on the $k\cdot\poly(n)$ points in $\mathbb{R}^d$ with nonzero support in the $k$ distributions and connect an edge between each pair of points if their pairwise distance is distorted by the random projection map $\pi$ by at least a $(1+\eps)$-factor.
To resolve issue (1), we give a combinatorial argument that shows that the distortion graph for $\pi$ is everywhere sparse for a weighted notion of sparsity. 
To resolve issues (2) and (3), we define a mapping for each point in a probability distribution that partitions its mass among the barycenters. 
Using the everywhere-sparse distortion graph, we show a robust $1$-point extension theorem that the pairwise distances from the barycenter to a large \emph{weighted} fraction of the points is preserved. 
Finally to resolve issue (4), we show that the remaining weighted fraction of points incurs a cost that is at most $\eps$-fraction of the optimal cost induced by the Wasserstein barycenter.

The structure of the proof is as follows. To prove that the cost of any solution (in the sense of Definition \ref{def:solution}) is preserved, we first show that the cost of the flow from a weighted cluster of points in the $\mu_i$s to one particular support point in the barycenter is preserved (Theorem~\ref{thm:cluster:preserve}). This in turn rests on the fact that weighted cluster costs are preserved when only a small weighted fraction of the cluster distances are distorted (Theorem~\ref{thm:cost:distortion}).

In summary, the overall proof structure is
$$ \textrm{Theorem \ref{thm:cost:distortion}} \Longrightarrow \textrm{Theorem \ref{thm:cluster:preserve}} \Longrightarrow \textrm{Theorem \ref{thm:wb:main}} \Longrightarrow \textrm{Theorem \ref{thm:jl:main}.} $$

We begin by proving Theorem~\ref{thm:cost:distortion}, which should be considered the weighted analog to Theorem 3.3 in \cite{MakarychevMR19}. 
\begin{restatable}{theorem}{thmcostdistortion}
\label{thm:cost:distortion}
Let $X\subset\mathbb{R}^d$ be a finite set of weighted points and the map $\phi:X\to\mathbb{R}^m$ have a distortion graph $G$ for $X$ that is $\alpha$-sparse (with respect to the weight of $X$), with $\alpha\le1/10^{p+1}$. 
Then for every $p\ge 1$ and $D=(1+\eps)^p(1+3^{p+2}+\alpha^{1/(p+1)})$, we have
\[\frac{1}{D}\C_p(X)\le\C_p(\phi(X))\le D\C_p(X)\]
where $\C_p$ is the cost of solving a clustering on $X$ with only $1$ center under the cost function $\| \cdot \|_2^p$.
\end{restatable}
Unfortunately, the results of \cite{MakarychevMR19} do not immediately imply the corresponding sparsity results for weighted graphs. 
For example, a vertex that has edges to a small fraction of its neighbors may still have an edge to a large weighted fraction of its neighbors. 
Thus we show the weighted analogs of the structural results from \cite{MakarychevMR19}. 
The following lemma is analogous to Lemma 4.1 in \cite{MakarychevMR19}, extending the properties to handle weighted sets $X$. 
\begin{lemma}
\label{lem:weighted:measure}
Let $X$ be a finite set and $V\subset X$ be a random subset of $X$. 
Let $\alpha\in(0,1/2)$ and suppose that $\PPr{x\in V}\ge2\alpha$ for each $x\in X$. 
Then there exist a random set $R\subset V$ and a deterministic measure $\mu$ on $X$ such that
\begin{enumerate}
\item 
$\mu(x)\ge\frac{w(x)}{w(V\setminus R)}$ for every $x\in V\setminus R$
\item
$\PPr{x\in R}\le2\alpha$ for every $x\in X$
\item
$\mu(X)=\sum_{x\in X}\mu(x)\le\frac{\PPr{V\neq\emptyset}}{\alpha^2}$
\end{enumerate}
\end{lemma}
\begin{proof}
Since $X$ is a finite set, we truncate (or discretize) the weights of the elements in $X$ and without loss of generality suppose that there exists a sufficiently large integer $N>0$ such that for each $x\in X$, there exists some integer $i\le N$ such that $w(x)=\frac{i}{N}$. 
We then prove our claim by induction on the weight of the set $X$. 
If $w(X)=0$ so that $X$ is empty, then the claim trivially holds. 
Now we suppose that $w(X)=\frac{k}{N}$ and the statement holds for all sets $X'$ with weight $w(X')=\frac{k'}{N}$, where $k'<k$ are non-negative integers; we show the statement holds for $X$. 

Let $\ell=\alpha\cdot w(X)$ and define a deterministic set $X'$ and a random subset $V'\subset X'$ by:
\begin{align*}
X'&=\{x:\PPr{x\in V\text{ and }w(V)<\ell}\ge2\alpha\}\\
V'&=\begin{cases}V\cap X',\qquad&\text{if }w(V)<\ell\\
\emptyset,\qquad&\text{otherwise}\end{cases}
\end{align*}

We first show that there exists an $x_0\in X$ such that $\PPr{x_0\in V\text{ and }w(V)<\ell}\le\alpha$, which implies that $x_0\notin X'$ and thus $w(X')<w(X)$. 
We show that the average value of $\PPr{x\in V\text{ and }w(V)<\ell}$ for $x\in X$ is at most $\alpha$, which implies the existence of such an $x_0$. 
Since $w(V)\cdot\mathds{1}\{w(V)<\ell\}$ is always at most $\ell$, then
\begin{align*}
\frac{1}{w(X)}\sum_{x\in X}&\PPr{x\in V\text{ and }w(V)<\ell}=\frac{1}{w(X)}\sum_{x\in X}\Ex{\mathds{1}\{x\in V\text{ and }w(V)<\ell\}}\\
&=\frac{1}{w(X)}\Ex{\sum_{x\in X}\mathds{1}\{x\in V\text{ and }w(V)<\ell\}}=\frac{1}{w(X)}\Ex{w(V)\cdot\mathds{1}\{w(V)<\ell\}}\le\frac{\ell}{w(X)}=\alpha.
\end{align*}

Because $w(X')<w(X)$ and $\PPr{x\in V'}=\PPr{x\in V\text{ and }w(V)<\ell}\ge2\alpha$ for each $x\in X'$ by definition of $X'$, then we apply the inductive hypothesis to $X'$ and $V'$. 
Hence, there exist a random set $R'\subset X'$ and a measure $\mu'$ on $X'$ such that the above claims 1-3 hold for $X'$ and $V'$. 
We then define a measure $\mu$ on $X$ and random subset $R\subset V$ by:
\begin{align*}
\mu(x)&=\begin{cases}\mu'(x)+\frac{w(x)}{\ell},\qquad&\text{if }x\in X'\\
\frac{w(x)}{\ell},\qquad&\text{otherwise}.\end{cases}\\
R&=\begin{cases}R'\cup(V\setminus X'),\qquad&\text{if }w(V)<\ell\\
R',\qquad&\text{otherwise}.\end{cases}
\end{align*}
We claim that $R$ and $\mu$ satisfy the desired properties.

\medskip\noindent
\textbf{Property 1: $\mu(x)\ge\frac{w(x)}{w(V\setminus R)}$ for each $x\in V\setminus R$.} 
Let $x\in V\setminus R$. 
\\
\noindent We have three possible cases. 
(1) If $x\in X'$ and $w(V)<\ell$, then $V\setminus R=V'\setminus R'$ by the definition of $R$. 
Hence, $\mu(x)>\mu'(x)\ge\frac{w(x)}{w(V'\setminus R)}=\frac{w(x)}{w(V\setminus R)}$ by the inductive hypothesis. 
(2) If $x\in X'$ and $w(V)\ge\ell$, then $\mu(x)\ge\frac{w(x)}{\ell}\ge\frac{w(x)}{w(V)}$. 
Since $w(V)\ge\ell$, we also have $V'=\emptyset$, so by the definition of $R$, we have $R=R'\subset V'=\emptyset$. 
Thus, $\frac{w(x)}{w(V)}=\frac{w(x)}{w(V\setminus R)}$ so that $\mu(x)\ge\frac{w(x)}{w(V\setminus R)}$. 
(3) If $x\notin X'$, then $x\in V\setminus X'\subset R'\cup(V\setminus X')$. 
Since $X\in V\setminus R$, then $x\notin R$. 
Thus, $R\neq R'\cup(V\setminus X')$. 
By the definitions of $\mu$ and $R$, we have that $w(V)\ge\ell$ and $\mu(x)=\frac{w(x)}{\ell}$, so that $\mu(x)\ge\frac{w(x)}{w(V)}$. 
Since $w(V)\ge\ell$, then $\mu(x)\ge\frac{w(x)}{w(V)}=\frac{w(x)}{w(V\setminus R)}$. 

\medskip\noindent
\textbf{Property 2: $\PPr{x\in R}\le2\alpha$.}
If $x\in X'$, then by the inductive hypothesis, $\PPr{x\in R}=\PPr{x\in R'}\le2\alpha$. 
If $x\notin X'$, then by the definitions of $R$ and $X'$ respectively, we have that $\PPr{x\in R}=\PPr{x\in V\text{ and }w(V)<\ell}\le2\alpha$. 

\medskip\noindent
\textbf{Property 3: $\mu(X)=\sum_{x\in X}\mu(x)\le\PPr{V\neq\emptyset}/\alpha^2$.} 
By the inductive hypothesis, $\mu'(X')\le\PPr{V'\neq\emptyset}/\alpha^2$. 
Therefore, $\mu(X)=\mu'(X')+\frac{w(X)}{\ell}\le\frac{\PPr{V'\neq\emptyset}}{\alpha^2}+\frac{1}{\alpha}$. 
Note that $\PPr{V\neq\emptyset}-\PPr{V'\neq\emptyset}\ge\alpha$ implies $\frac{\PPr{V'\neq\emptyset}}{\alpha^2}+\frac{1}{\alpha}\le\frac{\PPr{V\neq\emptyset}}{\alpha^2}$, which would imply the desired claim; thus it suffices to prove $\PPr{V\neq\emptyset}-\PPr{V'\neq\emptyset}\ge\alpha$. 

Observe that if $w(V)\ge\ell$, then $V\neq\emptyset$ but $V'=\emptyset$. 
Since $V'\subset V$, then
\[\PPr{V\neq\emptyset}-\PPr{V'\neq\emptyset}=\PPr{V\neq\emptyset\text{ and }V'=\emptyset}\ge\PPr{w(V)\ge\ell}.\]
Recall that we previously showed the existence of an $x_0\in X$ with $\PPr{x_0\in V\text{ and }w(V)<\ell}\le\alpha$. 
On the other hand, $\PPr{x\in V}\ge2\alpha$ for all $x\in X$. 
Thus,
\[\PPr{w(V)\ge\ell}\ge\PPr{x_0\in V\text{ and }w(V)\ge\ell}\ge\PPr{x_0\in V}-\PPr{x_0\in V\text{ and }w(V)<\ell}\ge\alpha.\]
Hence, we have shown that all three properties are satisfied by $R$ and $\mu$, which completes the induction for $X$. 
\end{proof}

The following claim is analogous to Corollary 4.2 in \cite{MakarychevMR19}, again extending the properties to handle weighted sets $X$. 
\begin{corollary}
\label{cor:weighted:measure}
Let $X$ be a finite set, $V\subset X$ be a random subset of $X$, and $\alpha\in(0,1/2)$. 
Then there exist a random set $R\subset V$ and a measure $\mu$ on $X$ such that
\begin{enumerate}
\item 
$\mu(x)\ge\frac{w(x)}{w(V\setminus R)}$ for every $x\in V\setminus R$
\item
$\PPr{x\in R}\le2\alpha$ for every $x\in X$
\item
$\mu(X)=\sum_{x\in X}\mu(x)\le\frac{1}{\alpha^2}$
\end{enumerate}
\end{corollary}
\begin{proof}
Let $X'=\{x:\PPr{x\in V}\ge2\alpha$. 
By applying Lemma~\ref{lem:weighted:measure} to $X'$ and $V'=V\cap X'$ and set
\begin{align*}
R&=R'\cup(V\setminus X'),&\\
\mu(x)&=\begin{cases}\mu'(x),\qquad&\text{if }x\in X'\\
0,\qquad&\text{otherwise}\end{cases}.
\end{align*}
\end{proof}

We also have the following analog to Observation 4.3 in \cite{MakarychevMR19}. 
\begin{observation}\label{obs:weight}
Let $R$ be defined as in Corollary~\ref{cor:weighted:measure} and $V_0=V\setminus R$. 
Then for every $S\subset V_0$, we have $w(S)\le\mu(S)\cdot w(V_0)$.
\end{observation}
\begin{proof}
By Corollary~\ref{cor:weighted:measure}, we have $\mu(x)\ge w(x)/w(V_0)$ for every $x\in S$. 
Thus, $\mu(S)\ge w(S)/w(V_0)$, so $w(S)\le\mu(S)\cdot w(V_0)$. 
\end{proof}

The following is analogous to Theorem $3.2$ in
\cite{MakarychevMR19}.
\begin{theorem}[Theorem $3.2$ in \cite{MakarychevMR19}]\label{thm:sparse_graph}
Consider a finite set $X$ and a random graph $H = (V, E)$, where $V$ is a random subset of $X$ and $E$ is a random set of edges between vertices in $V$ (there are no independence assumptions or any other implicit assumptions about the distribution of $V$ and $E$). Let $\alpha \in (0, 1/2)$. Assume that $\PPr{(x, y)\in E} \le \delta \le \alpha$ for every $x,y \in X$. Then there exists a random subset $V' \subset V$ such that
\begin{itemize}
    \item $H[V']$ is $\alpha$-everywhere sparse,
    \item $\PPr{u \in V \setminus V'} \le 600 \delta/\alpha^6$ for all $u \in X$.
\end{itemize}
\end{theorem}

\begin{proof}
Let $\beta = \alpha/(1+\alpha)$ and $\alpha' = \alpha/3$. Applying Corollary \ref{cor:weighted:measure} with $\alpha'$, we get a deterministic measure $\mu$ on $V$ and a random set $R \subset V$. Consider the canonical product measure 
\[ \mu^{\otimes 2}((x,y)) = \mu(x)\mu(y) \]
for all $x,y \in X$. We define $V'$ according to the measure $\mu^{\otimes 2}$.
\\

\noindent \textbf{Case 1}: $\mu^{\otimes 2}(E) \ge \beta^2$. In this case, we let $V' = \emptyset$.
\\

\noindent \textbf{Case 2}:  $\mu^{\otimes 2}(E) < \beta^2$. In this case, define $V_0 = V \setminus R$. We say that $x \in X$ is bad if $\mu(\{y \in V_0: (x, y) \in E\}) \ge \beta$. Let $B$ denote the set of bad vertices and define 
\[V' = V \setminus (R \cup B)  = V_0 \setminus B. \]

Our goal is to verify that in both of the above cases, both of the desired properties of Theorem \ref{thm:sparse_graph} hold. First we handle the easier Case $1$. There, the graph $H[V']$ is empty so the conclusion trivially holds. For the second condition, note that
\[\PPr{ \mu^{\otimes 2}(E) \ge \beta^2} \le \frac{\Ex{\mu^{\otimes 2}(E)}}{\beta^2} \le \frac{\delta M^2}{\beta^2} \]
where $M = \mu(X)$ and we have used the fact that $\PPr{(x, y) \in E} \le \delta$. Then from our choice of $\beta$ and $\alpha'$ (fill in details in a bit), the above probability is at most $C \alpha$.

We now verify the second case. First, we check that $H[V']$ is $\alpha$-everywhere sparse. This is equivalent to checking that the weighted degree of every vertex $x$ in $H[V']$ is at most $\alpha$ fraction of the total weight. That is, we need to check:
\[ w( \{ y \in V': (x,y) \in E \}) \le \alpha w(V'). \]
Analogous to the proof of Theorem $3.2$ in \cite{MakarychevMR19}, we have $\mu^{\otimes 2}(E) \ge \beta \mu(B)$ and therefore, $\mu(V_0 \cap B) \le \mu(B) \le \beta$ since we are in Case $2$. Now for $x \in V'$, we similarly have $\mu(y \in V' : (x,y) \in V') \le \beta$. Combining the above findings with Observation \ref{obs:weight}, we conclude the following two statements:
\begin{itemize}
    \item $w(\{ y \in V': (x,y) \in E \}) \le \beta w(V_0)$,
    \item $w(V_0 \cap B) \le \beta w(V_0)$.
\end{itemize}
From the second relation, we have $w(V') = w(V_0 \setminus B) \ge (1-\beta)w(V_0) $
and using the first relation, we can conclude that
\[  w(\{ y \in V': (x,y) \in E \}) \le \frac{\beta}{1-\beta} w(V').\]
Finally, the same probability bound as in the end of the proof of Theorem $3.2$ in \cite{MakarychevMR19} allows us to say that in Case $2$, the probability of $x \in B$ is at most $\delta M / \beta$. Finally, combining all the probabilities from Case $1$ and Case $2$, we can conclude identically as in Theorem $3.2$ in \cite{MakarychevMR19} that $\PPr{u \in V \ V'} \le \delta M^2/\beta^2 + 2 \alpha' + \delta M /\beta \le 600 \delta / \alpha^6$.
\end{proof}

We now state an analogous version of Theorem $5.2$ from \cite{MakarychevMR19} that is suitable for our purposes. 

\begin{theorem}[Robust Kirszbraun Theorem]
\label{thm:robust} 
Consider two finite (multi) sets of points $X\subset \R^d$ and $Y\subset \R^{m}$ and a map $\varphi:X \rightarrow Y$. Let $G =(X, E)$ be the distance expansion graph for $\varphi$ with respect to the Euclidean distance with vertex weights given by $w: X \rightarrow \R^{\ge 0}$. Suppose that $G$ is $\alpha$-everywhere sparse according to Definition \ref{def:sparse}. Then for every $u\in\R^d$ and $\eps>0$, there exists $v\in\R^m$ and $X' \subset X$ such that for all $x \in X \setminus X'$,
\[\|\phi(x)-v\|\le(1+\eps)\|x-u\|,\]
and $w(X') \le \alpha'(\eps) w(X)$ where $\alpha'(\eps)=2(1+\eps)^2\alpha/\eps$. 
\end{theorem}
The proof of Theorem $5.2$ from \cite{MakarychevMR19} carries over in a straightforward fashion to the proof of Theorem \ref{thm:robust} above. In particular, we just outline the small changes that need to occur to carry the proof over. 
\begin{proof}[Proof Sketch]
In \cite{MakarychevMR19}, the following polytope is defined:
\[\Lambda_{\eta} = \{ \lambda \in \R^X : \sum_{x \in X} \lambda_x = 1 ; 0 \le \lambda_{x'} \le \eta \text{ for all } x' \in X \} \]
for $\eta = (\alpha'(\varepsilon)n)^{-1}$. For us, we define a slightly modified polytope which includes the weights of elements of $X$:
\[\Lambda_{\eta} = \{ \lambda \in \R^X : \sum_{x \in X} \lambda_x = 1 ; 0 \le \lambda_{x'} \le w(x')\eta \text{ for all } x' \in X \} \]
where $\eta = (\alpha'(\varepsilon) w(X))^{-1}$. Then for every $\lambda \in \Lambda_{\eta}, u' \in \R^{d'}$, and $v' \in R^{d''}$, we similarly let 
\[ f(X, \lambda, u') = \sum_{x \in X} \lambda_x \|u'-x\|^2 \text{ and } f(\phi(X), \lambda, v') = \sum_{x \in X} \lambda_x \|v' - \phi(x) \|^2. \]
Since $\Lambda_{\eta}$ is still a convex polytope, we recover the statement
\begin{equation}\label{eq:functional}
    \max_{v' \in \R^{d''}} \min_{\lambda \in \Lambda_{\eta}} F(v', \lambda) \ge 0 
\end{equation}
where $F(v', \lambda) = (1+ \varepsilon)^2 f(X, \lambda, u)-f(\varphi(X), \lambda, v')$. Now to finish the rest of the proof, let $v$ be the point that maximizes the functional $\min_{\lambda \in \Lambda_{\eta}} F(v, \lambda)$. By \eqref{eq:functional}, we know that $F(v, \lambda) \ge 0$ for all $\lambda \in \Lambda_{\eta}$. Now consider the set 
\[ S = \{x \in X : \|\varphi(x) - v \| \ge (1+\varepsilon) \|x-u\| \}. \]
If $S = \emptyset$, we are done so otherwise, define $\lambda^*$ as 
\[ \lambda_x^* = \begin{cases} \frac{w(x)}{w(S)} &\mbox{if } x \in S, \\
0 &\mbox{otherwise }.\end{cases} \]
By the definition of $S$, we have $(1+\varepsilon)\|x-u\|^2 - \|\varphi(x)-v \|^2 < 0$ which implies
\[ F(v, \lambda^*) = \frac{1}{w(S)}\sum_{x \in S}w(x) \left( (1+\varepsilon)\|x-u\|^2 - \|\varphi(x)-v \|^2\right) < 0 \]
and thus, $\lambda^* \not \in \Lambda_{\eta}$. Therefore, $1/w(S) > \eta$ and $w(S) < 1/\eta = \alpha'(\varepsilon) w(X)$. This finishes the proof of Theorem \ref{thm:robust}.
\end{proof}

We also need the following analogous version of Lemma $5.3$ in \cite{MakarychevMR19}. The proof differs in that we have to consider a careful weighting scheme whereas in \cite{MakarychevMR19} it was more straightforward.
We first require the following property:
\begin{lemma}[Lemma A.1 in \cite{MakarychevMR19}]
\label{lem:mmr:a1}
Let $x$ and $y_1,\ldots,y_r$ be non-negative real numbers, and $\eps>0$, $p\ge 1$. 
Then
\[\left(x+\sum_{i=1}^r y_i\right)^p\le (1+\eps)^{p-1}x^p+\left(\frac{(1+\eps)r}{\eps}\right)^{p-1}\sum_{i=1}^r y_i^p.\]
\end{lemma}

\begin{lemma}[Lemma $5.3$ in \cite{MakarychevMR19}]\label{lem:onecost}
Consider two finite multisets of points $X\subset \R^{d}$ and $Y \subset \R^{m}$ of the same size and a one-to-one map $\varphi:X \rightarrow Y$. Let $G=(X, E)$ be the distance expansion graph for $\varphi$ with respect to the Euclidean distance with a weight function $w : X \rightarrow \R^{\ge 0}$. Suppose that $G$ is $\alpha$-everywhere sparse with $\alpha \le 1/10^{p+1}$. Then, for every $p\ge 1$, we have the following inequality on the cost of the clusters $X$ and $Y$ (with the same weights as $w$) 
\[ \C_p(Y) \le (1 + 3^{p+2}\alpha^{1/(p+1)}) \C_p(X).\]
\end{lemma}
\begin{proof}
Let $\varepsilon = \alpha^{1/(p+1)}$ and let $u^*$ be the optimal center for the cluster $X$. By Theorem \ref{thm:robust}, there exists a set $\widetilde{X} \subset X$ and a point $v^* \in \R^m$ such that for $x \in \tilde{X}$,
\[ \| \varphi(x) - v^*\| \le (1+\varepsilon)\|x - u^*\| \]
where  $w(X \setminus \widetilde{X}) \le \alpha'w(X)$ with $\alpha' \le 2(1+\varepsilon)^2 \alpha / \varepsilon$. By definition, it follows that 
\begin{align*}
   \C_p(Y) &\le \sum_{y \in Y} w(y)\|y - v^* \|^p \\
   &= \sum_{x \in X} w(x)\|\varphi(x) - v^*\|^p \, ( \text{$x$ and $\varphi(x) = y$ have the same weight}) \\
   &= \sum_{x \in \widetilde{X}} w(x)\|\varphi(x) - v^*\|^p + \sum_{x \not \in \widetilde{X}} w(x)\|\varphi(x) - v^*\|^p \\
   &\le (1+\varepsilon)^p \sum_{x \in \widetilde{X}} w(x)\|x - u^*\|^p + \sum_{x \not \in \widetilde{X}} w(x)\|\varphi(x) - v^*\|^p.
\end{align*}
We now try to bound $ \|\varphi(x) - v^*\|^p$ for $x \not \in \widetilde{X}$. We will bound this quantity using a slightly stronger claim that applies for all $x \in X$. Indeed, fix an arbitrary $x$ and consider the set $I_x$ of its non neighbors in the distance expansion graph $G$. Note that $w(I_x) \ge (1-\alpha)w(X)$ and thus, $w(I_x \cap \widetilde{X}) \ge (1-\alpha - \alpha')w(X) > 0$ if the total weight $w(X)$ is positive. Consider an arbitrary $x' \in I_x \cap \widetilde{X}$. Then it follows that
\begin{align*}
    \|\varphi(x) - v^*\| &\le \|\varphi(x) - \varphi(x')\| + \|\varphi(x') - v^*\| \\
    &\le (1+\varepsilon)\|x-x'\| + (1+\varepsilon)\|x'-u^*\| \ (\text{using the fact that } x' \in I_x \cap \widetilde{X}) \\
    &= (1+\varepsilon)\|x-u^*\| + (2+2\varepsilon)\|x'-u^*\|.
\end{align*}
Applying Lemma \ref{lem:mmr:a1}, we have that 
\begin{equation}\label{eq:avg}
     \|\varphi(x) - v^* \| \le (1+\varepsilon)^p \|x - u^* \|^p + \frac{3^p}{\varepsilon^{p-1}} \|x' - u^*\|^p 
\end{equation}
for $\varepsilon$ sufficiently small. We now average equation \eqref{eq:avg} over all possible $x'$. This gives us 
\begin{align*}
    \|\varphi(x) - v^* \| &\le (1+\varepsilon)^p \|x - u^* \|^p + \frac{3^p}{\varepsilon^{p-1}} \cdot \frac{1}{w(I_x \cap \widetilde{X})}\sum_{x' \in I_x \cap \widetilde{X} } w(x') \, \|x' - u^*\|^p  \\
    &\le (1+\varepsilon)^p \|x - u^* \|^p + \frac{3^p}{\varepsilon^{p-1}} \cdot \frac{1}{\alpha''w(X)}\sum_{x' \in I_x \cap \widetilde{X} } w(x') \, \|x' - u^*\|^p 
\end{align*}
where $\alpha'' = (1-\alpha - \alpha')$. Therefore,
\begin{align*}
    &\sum_{x \not \in \widetilde{X}} w(x)\|\varphi(x) - v^*\|^p  \le (1+\varepsilon)^p \sum_{x \not \in \widetilde{X}} w(x)\|x-u^*\|^p \\
    &+  \frac{3^p}{\varepsilon^{p-1}} \cdot  \frac{1}{\alpha''w(X)} \sum_{\substack{x \not \in \widetilde{X}\\ x' \in I_x \cap \widetilde{X}} }w(x) w(x')\|x' - u^*\|^p.
\end{align*}
Focusing on the second term, we have
\begin{align*}
     \frac{1}{\alpha''w(X)} \sum_{\substack{x \not \in \widetilde{X}\\ x' \in I_x \cap \widetilde{X}} }w(x) w(x')\|x' - u^*\|^p &\le  \frac{1}{\alpha''w(X)} \sum_{x \not \in \widetilde{X}} w(x) \sum_{x' \in \widetilde{X}} \|x' - u^*\|^p \\
     &\le \frac{1}{\alpha''w(X)} \sum_{x \not \in \widetilde{X}} w(x) \C_p(X) \\
     &\le \C_p(X) \, \frac{w(X \setminus \widetilde{X})}{\alpha'' w(X)} \\
     &\le 2 \alpha ' \, \C_p(X)
\end{align*}
using the fact that $\alpha'' \ge 1/2$. Putting everything together gives us
\begin{align*}
    \C_p(Y) &\le (1+\varepsilon)^p \sum_{x \in X} w(x) \|x-u^*\|^p + \frac{3^p}{\varepsilon^{p-1}} \cdot 2 \alpha' \C_p(X) \\
    &\le (1+3^{p+2} \alpha') \C_p(X)
\end{align*}
for $\varepsilon$ sufficiently small.
\end{proof}

Finally, with the above lemmas in hand, we are ready to prove Theorem~\ref{thm:cost:distortion}.
\begin{proof}[Proof of Theorem~\ref{thm:cost:distortion}]
Given Lemma \ref{lem:onecost}, the proof follows identically as the proof of Theorem $3.2$ in \cite{MakarychevMR19} by applying Lemma \ref{lem:onecost} to the maps $(1+\varepsilon)\varphi$ and $(1+\varepsilon)\varphi^{-1}$.
\end{proof}

Next we prove Theorem~\ref{thm:cluster:preserve}, which shows that the cost of each ``cluster" in the barycenters problem (i.e. the cost of the weighted flow from points in the input distributions to \emph{one} of the support points in the barycenter) is preserved. We will make use of \emph{distortion graphs}, which quantify the level of distortion of pairwise distances resulting from a dimensionality reduction map.
\begin{definition}
Let $\pi: \R^d \rightarrow \R^m$ and $X$ be a set of points in $\R^d$. A distortion graph $G$ with vertex set $X$ is a graph where two points $u,v \in X$ are joined by an edge if the distance between $u$ and $v$ is distorted by a factor at least $1+\varepsilon$ by $\pi$.
\end{definition}

We define the following concept of an everywhere sparse graph to be a generalization to weighted graphs of the concept introduced by~\cite{MakarychevMR19}. 
\begin{definition}\label{def:sparse}
Let $G=(V,E)$ be a graph with vertex weights given by $w : V \rightarrow \R^{\ge 0}$. 
Let $N(u)$ denote the neighborhood of a vertex $u$. $G$ is $\alpha$-everywhere sparse if
\[ \sum_{v \in N(u)} w(v) \le \alpha w(V) = \alpha \, \sum_{v \in V} w(v) \]
for all $u \in V$.
\end{definition}

Finally, we impose some additional requirements on the dimensionality reduction map. These essentially say that, even when a pair of points are distorted by the reduction map, the distortion is not too large in expectation.
\begin{definition}
\label{def:dim:red}
For $\eps>0$, $\delta\in(0,1)$, a random map $\pi:\R^d\to\R^m$ is an $(\eps,\delta)$-dimension reduction if
\[\frac{1}{1+\eps}\|x-y\|\le\|\pi(x)-\pi(y)\|\le(1+\eps)\|x-y\|,\]
with probability at least $1-\delta$ for every $x,y\in\R^d$. 
For $p\ge1$, $\pi$ is an $(\eps,\delta,\alpha)$-dimension reduction if it additionally satisfies
\[\Ex{\mathds{1}\{\|\pi(x)-\pi(y)\|>(1+\eps)\|x-y\|\}\left(\frac{\|\pi(x)-\pi(y)\|^p}{\|x-y\|^p}-(1+\eps)^p\right)}\le\alpha.\]
We say $\pi$ is a standard dimension reduction if the parameters $(\eps,\delta,\alpha)$ permit $\delta\le\exp(-C\eps^2d)$ and $\alpha\le\exp(-C\eps^2d)$ for $d\ge C'p/\eps^2$ for some absolute constants $C,C'>0$. 
\end{definition}

In the theorem below, we just consider one point $\v^j$ in the support of the barycenter and corresponding points (and weights) $S_j, w_j(\cdot)$ assigned to $\v^j$. 
We show that a random $(\eps,\delta,\alpha)$-standard dimensionality reduction map roughly preserves the cost of the assignment. 
The proof is similar to Theorem 3.4 in \cite{MakarychevMR19}, but we instead use a weighted version of the distortion graphs.

\begin{restatable}{theorem}{thmclusterpreserve}
\label{thm:cluster:preserve}
Let $\mu_1, \ldots, \mu_k$ be an instance of the Wasserstein barycenter problem with the $L_p$ objective. Define $X = \bigcup_{i=1}^k \mathrm{supp}(\mu_i)$. Let $\pi$ be a random $(\eps,\delta,\alpha)$-standard dimensionality reduction map and let $\v_* = \sum_{j=1}^n b_j \d(c^*_j)$ be the Wasserstein barycenter. Furthermore, let $\calC^* = (C^*_1, \ldots, C^*_n)$ denote the solution to the minimum flow problem for the $\mu_i$s to $\v^*$ with corresponding weight functions $w^*_j(\cdot)$ (see Definition \ref{def:solution}).

Let $\calC = (C_1, \ldots, C_n)$ with corresponding weight functions $w_j(\cdot)$ be any solution (possibly random that depends on $\pi$) to the Wasserstein barycenter problem in the sense of Definition \ref{def:solution}. Let $C = C_j$ be any fixed cluster in $\calC$, and further suppose that $\alpha\le1/10^{p+1}$ and $\delta\le\min(\alpha^7/600,\alpha/n)$. 
Then with probability at least $1-\eta-\binom{n}{2}\delta$,
\begin{align*}
\C_p(\pi(C))&\le A(\C_p(C)+c_{\eta\eps\alpha}\C_p(\calC^*))\\
\C_p(C)&\le A(\C_p(\pi(C))+c_{\eta\eps\alpha}\C_p(\calC^*)),
\end{align*}
where $A=(1+\eps)^{3p-2}(1+3^{p+2}\alpha^{1/(p+1)})$ and $c_{\eta\eps\alpha}=\frac{5(1+\eps)^p\alpha}{\eta\eps^{p-1}}$. 
\end{restatable}

\begin{proof}
We use a similar outline to the proof of Theorem 3.4 in \cite{MakarychevMR19}. 
Fix $C = C_j$ and let $w_C(\cdot) \equiv w_j(\cdot)$.
Let $\calE$ be the event that all distances between the points $c^*_i$ are preserved within a $(1+\eps)$-approximation, so that $\PPr{\calE}\ge 1-\binom{n}{2}\delta$ under a random JL projection (or other random standard dimensionality reduction projection). 
We thus condition the remainder of the proof on the event $\calE$. 
Let $C^\circ\subset C$ be the subset of $C$ whose distances to each center $c^*_i$ are preserved within a $(1+\eps)$-approximation, so that
\[C^\circ=\{x\in C:\pi\text{ preserves the distance between }x\text{ and each }c^*_i\text{ within a factor of }(1+\eps)\}.\]
Note that for a particular $x\in X$, we have that $\PPr{x\in C\setminus C^\circ}\le n\delta$, since by a union bound, the probability that the distance between $x$ and some center $c^*_i$ is distorted by more than a $(1+\eps)$-approximation is at most $n\delta$. 
Let $G[C^\circ]$ be the graph induced by $C^\circ$ on the distortion graph $G$. 
By Theorem~\ref{thm:sparse_graph}, there exists a set $C'\subset C^\circ$ such that $G[C']$ is $\alpha$-everywhere sparse and $\PPr{x\in C^\circ\setminus C'}\le\alpha$. 
Thus,
\[\PPr{x\in C\setminus C'}\le\PPr{x\in C\setminus C^\circ}+\PPr{x\in C^\circ\setminus C'}\le n\delta+\alpha\le2\alpha.\]
We define $g(x)$ to be the identity mapping if $x\in C'$ (so that $x$ is a vertex of the $\alpha$-everywhere sparse graph) and otherwise, we define $g(x)$ to be the weighted multiset of the assignment of $x$ to each point in the support of the optimal Wasserstein barycenter $\calC^*$:
\[g(x)=\begin{cases}(w_C(x),x),\qquad\text{if }x\in C'\\
\{(w_C(x)\cdot r_i(x),c^*_i)\}_{i\in[n]},\qquad\text{if }x\notin C',
\end{cases}
\]
where $r_i(x)$ is the ratio of the weight of $x$ that is assigned to the point $c^*_i$. That is, $r_i(x) = w^*_i(x)/a(x)$, where recall $a(x) = $ weight of $x$ in $\mu_i$. Note that since $\sum_{i=1}^n w^*_i(x) = a(x)$ in order for $w^*_i(\cdot)$ to define a valid solution to the min flow problem, we have $\sum_{i=1}^n r_i(x) = 1$.

Let $\tilde{C}=g(C)$ be a multiset, so that every weighted point in $\tilde{C}$ is either assigned to a point in $C'$ or assigned to some point(s) $c^*_i$ in $\calC^*$. 
Let $\tilde{c}$ be the optimal center for $\tilde{C}$. 
That is, $$\tilde{c} = \argmin_c \sum_{(w, y) \in \tilde{C}} w\|y - c\|^p.$$
Observe that since $G[C']$ is $\alpha$-everywhere sparse, then the map $\pi$ $(1+\eps)$-approximates the distances from every $x\in C'$ to (1) at least a $(1-\alpha)$ weighted fraction of the points in $C'$ and (2) to all of the points in the barycenter $c^*_i$. 
Conditioning on the event $\mathcal{E}$ so that all pairwise distances between the points $c^*_i$ are approximated within a $(1+\eps)$ factor, then by Theorem~\ref{thm:cost:distortion},
\begin{align}
\label{eqn:distort}
\frac{1}{D}\C_p(\tilde{C})\le\C_p(\pi(\tilde{C}))\le D\C_p(\tilde{C}),
\end{align}
where $D=(1+\eps)^p(1+3^{p+2}\alpha^{1/(p+1)})$. 
Because $\tilde{c}$ is the optimal center for $\tilde{C}$, we further have
\[\C_p(C)\le\sum_{x\in C}w_C(x)\|x-\tilde{c}\|^p,\qquad\C_p(\tilde{C})=\sum_{x\in C}\sum_{y\in g(x)}w\|y-\tilde{c}\|^p,\]
where we denote each ordered pair in $g(x)$ as $(w,y)$. 
We compare each term in the right hand side of the relationships for $\C_p(C)$ and $\C_p(\tilde{C})$. 
For $x\in C'$, we have that $g(x)=x$ and $w=w_C(x)$ so the contributions of each term in both summations are the same. 
For $x\notin C'$, the contributions are $w_C(x)\|x-\tilde{c}\|^p$ and $w_C(x)\sum_{i=1}^n r_i(x)\|c^*_i-\tilde{c}\|^p$ respectively. 
Hence,
\[\C_p(C)-(1+\eps)^{p-1}\C_p(\tilde{C})\le\sum_{x\in C\setminus C'} \left(w_C(x)\|x-\tilde{c}\|^p-(1+\eps)^{p-1}w_C(x)\sum_{i=1}^n r_i(x)\|c^*_i-\tilde{c}\|^p\right).\]
By triangle inequality, we have that $\|x-\tilde{c}\|\le\|x-c^*_i\|+\|c^*_i-\tilde{c}\|$ for each $i$. 
By Lemma~\ref{lem:mmr:a1} with $r=1$,
\begin{equation}\label{eq:dist to pseudocenter}r_i(x)\|x-\tilde{c}\|^p\le(1+\eps)^{p-1}r_i(x)\|c^*_i-\tilde{c}\|^p+\left(\frac{1+\eps}{\eps}\right)^{p-1}r_i(x)\|x-c^*_i\|^p\end{equation}
for each $i$. Recalling that $\sum_{i=1}^n r_i(x) = 1$ and summing over inequality \eqref{eq:dist to pseudocenter}, we have
\[ \|x-\tilde{c}\|^p \leq (1 + \eps)^{p-1} \sum_{i=1}^n r_i(x) \| c^*_i - \tilde{c} \|^p + \left(\frac{1+\eps}{\eps}\right)^{p-1}\sum_{i=1}^n r_i(x) \| x - c^*_i \|^p. \]
Thus,
\[\C_p(C)-(1+\eps)^{p-1}\C_p(\tilde{C})\le\left(\frac{1+\eps}{\eps}\right)^{p-1}\sum_{x\in C\setminus C'}w_C(x)\sum_{i=1}^n r_i(x)\|c^*_i-x\|^p.\]
By similar reasoning, we also have
\begin{align}
\C_p(\tilde{C})-(1+\eps)^{p-1}\C_p(C)&\le\left(\frac{1+\eps}{\eps}\right)^{p-1}\sum_{x\in C\setminus C'}w_C(x)\sum_{i=1}^n r_i(x)\|c^*_i-x\|^p \label{eq:Ctilde vs C}\\
\C_p(\pi(C))-(1+\eps)^{p-1}\C_p(\pi(\tilde{C}))&\le\left(\frac{1+\eps}{\eps}\right)^{p-1}\sum_{x\in C\setminus C'}w_C(x)\sum_{i=1}^n r_i(x)\|\pi(c^*_i)-\pi(x)\|^p \label{eq:pi(C) vs pi(Ctilde)}\\
\C_p(\pi(\tilde{C}))-(1+\eps)^{p-1}\C_p(\pi(C))&\le\left(\frac{1+\eps}{\eps}\right)^{p-1}\sum_{x\in C\setminus C'}w_C(x)\sum_{i=1}^n r_i(x)\|\pi(c^*_i)-\pi(x)\|^p.
\end{align}
Along with (\ref{eqn:distort}), we have that
\begin{align*}
    \C_p(\pi(C)) &= \left[\C_p(\pi(C)) - (1+\eps)^{p-1}\C_p(\pi(\tilde{C}))\right] + (1+\eps)^{p-1}\C_p(\pi(\tilde{C})) \\[10pt]
    &\leq \eqref{eq:pi(C) vs pi(Ctilde)}  + (1+\eps)^{p-1} D \C_p(\tilde{C}) \qquad \text{(due to \eqref{eqn:distort})} \\[10pt]
    &= \eqref{eq:pi(C) vs pi(Ctilde)} + (1+\eps)^{p-1} D \left[ \left(\C_p(\tilde{C}) - (1+\eps)^{p-1} \C_p(C)\right) + (1+\eps)^{p-1} \C_p(C) \right] \\[10pt]
    &\leq \eqref{eq:pi(C) vs pi(Ctilde)} + (1+\eps)^{p-1} D \left[ \eqref{eq:Ctilde vs C} + (1+\eps)^{p-1} \C_p(C) \right].
\end{align*}
Substituting the right-hand sides of \eqref{eq:pi(C) vs pi(Ctilde)} and \eqref{eq:Ctilde vs C} into the above and factoring out $A = (1+\eps)^{2(p-1)}$, if we define
\begin{equation}\label{eq: Rx}
    R_x=w_C(x)\sum_{i=1}^n \left( \underbrace{r_i(x)\|c^*_i-x\|^p}_{\mathrm{(A)}}+\underbrace{r_i(x)\|\pi(c^*_i)-\pi(x)\|^p}_{\mathrm{(B)}} \right)
\end{equation}
we see that
\begin{equation} \label{eq: cost(pi(C)) bound}
\C_p(\pi(C))\le A\left(\C_p(C)+\eps^{1-p}\sum_{c\in C\setminus C'}R_x\right).
\end{equation}
A similar calculation yields
\begin{equation} \label{eq: cost(C) bound}
\C_p(C)\le A\left(\C_p(\pi(C))+\eps^{1-p}\sum_{c\in C\setminus C'}R_x\right),
\end{equation}
and inequalities \eqref{eq: cost(pi(C)) bound} and \eqref{eq: cost(C) bound} simultaneously hold with probability at least $1 - \binom{n}{2}\d$.

Finally, we prove that $\eps^{1-p}\sum_{x\in C\setminus C'}\mathds{1}\{\calE\}R_x\le c_{\eta\eps\alpha}\C_p(\calC^*)$ with probability at least $1-\eta$, by first showing that $\Ex{\mathds{1}\{\calE\}\sum_{x\in C\setminus C'} R_x}\le5(1+\eps)^p\alpha\C_p(\calC^*)$ and then applying Markov's inequality.
We will bound the sum of the (B) terms first. 
Observe that 
$$\| \pi(c^*_i) - \pi(x) \|^p \leq (1+\eps)^p\|c^*_i - x\|^p + \max(\|\pi(c^*_i)-\pi(x)\|^p-(1+\eps)^p\|c^*_i-x\|^p,0).$$
Furthermore, we have
\begin{align*}
    \max(\|\pi(c^*_i)&-\pi(x)\|^p-(1+\eps)^p\|c^*_i-x\|^p,0) = \\[10pt]
    &\I\{\|\pi(c^*_i)-\pi(x)\|>(1+\eps)\|c^*_i-x\|\}\left(\|\pi(c^*_i)-\pi(x)\|^p-(1+\eps)^p\|c^*_i-x\|^p\right).
\end{align*}
Thus by Definition~\ref{def:dim:red}, we have that for every $x\in C$ and any $i\in[n]$,
\[\Ex{\max(\|\pi(c^*_i)-\pi(x)\|^p-(1+\eps)^p\|c^*_i-x\|^p,0)}\le\alpha\|c^*_i-x\|^p.\]
Combining these two bounds, we see that
\begin{align*}
\Ex{\sum_{x\in C\setminus C'}w_C(x)\sum_{i=1}^n r_i(x)\,\|\pi(c^*_i)-\pi(x)\|^p} &\le(1+\eps)^p\,\Ex{\sum_{x\in C\setminus C'}w_C(x) \sum_{i=1}^n r_i(x) \|c^*_i-x\|^p} \\
&\hspace{-.5in}+ \Ex{\sum_{x\in C\setminus C'} w_C(x) \sum_{i=1}^n r_i(x) \max(\|\pi(c^*_i)-\pi(x)\|^p-(1+\eps)^p\|c^*_i-x\|^p,0)}.
\end{align*}
The second term in the RHS can be bounded by
\begin{align*}
   &\Ex{\sum_{x\in C\setminus C'}w_C(x)\sum_{i=1}^n r_i(x)\max(\|\pi(c^*_i)-\pi(x)\|^p-(1+\eps)^p\|c^*_i-x\|^p,0)} \\
   &\le\sum_{x\in X}\sum_{i=1}^n w_C(x) r_i(x) \Ex{\max(\|\pi(c^*_i)-\pi(x)\|^p-(1+\eps)^p\|c^*_i-x\|^p,0)} \\
   &\leq\alpha \sum_{x \in X} \sum_{i=1}^n w^*_i(x) \|c^*_i-x\|^p \\
   &\le \alpha\,\C_p(\calC^*),
\end{align*}
where we have replaced the sum over $C\setminus C'$ by a larger sum over all $X$ and used the fact that $w_C(x)r_i(x)\le a(x)r_i(x)=w^*_i(x)$ for the weight $w^*_i(x)$ that is assigned to the point $c^*_i$ in the actual barycenter. 
Therefore,
\begin{align}
\Ex{\sum_{x\in C\setminus C'}w_C(x)\sum_{i=1}^n r_i(x)\,\|\pi(c^*_i)-\pi(x)\|^p}&\le(1+\eps)^p\,\Ex{\sum_{x\in C\setminus C'}w_C(x)\sum_{i=1}^n r_i(x)\,\|c^*_i-x\|^p} \label{eq:bounding (B) terms}\\
&+\alpha\,\C_p(\calC^*). \nonumber
\end{align}
By linearity of expectation, we also have
\begin{align*}
\Ex{\sum_{x\in C\setminus C'}w_C(x)\|x-c^*_i\|^p} &= \Ex{\sum_{x\in X} \I\{x \in C\setminus C'\} w_C(x)\|x-c^*_i\|^p} \\
&= \sum_{x\in X}\PPr{x\in C\setminus C'}\,w_C(x)\|x-c^*_i\|^p \\
&\le2\alpha\sum_{x\in X}w_C(x)\|x-c^*_i\|^p.
\end{align*}
It then follows that
\begin{align}
    \Ex{\sum_{x\in C\setminus C'}w_C(x)\sum_{i=1}^n r_i(x)\|x-c^*_i\|^p} &= \sum_{i=1}^n r_i(x)\Ex{\sum_{x\in C\setminus C'} w_C(x) \| x - c^*_i \|^p} \nonumber \\
    &\leq \sum_{i=1}^n r_i(x) \cdot 2\alpha\sum_{x\in X} w_C(x) \|x - c^*_i\|^p \nonumber \\
    &\leq 2\alpha \sum_{i=1}^n \sum_{x\in X} w^*_i(x) \|x - c^*_i\|^p \nonumber \\
    &= 2\alpha\,\C_p(\calC^*) \label{eq:bounding (A) and part of (B)}
\end{align}
where we again use the fact that $w_C(x)r_i(x)\le a(x)r_i(x)=w^*_i(x)$. Notice that the bound obtained by \eqref{eq:bounding (A) and part of (B)} suffices to bound the first term in \eqref{eq:bounding (B) terms} \emph{and} the (A) terms from \eqref{eq: Rx}.
Combining bounds \eqref{eq:bounding (B) terms} and \eqref{eq:bounding (A) and part of (B)} with the definition of $R_x$ in \eqref{eq: Rx}, we obtain
\begin{align*}
\Ex{\mathds{1}\{\calE\}\sum_{x\in C\setminus C'}R_x}&=\Ex{\sum_{x\in C\setminus C'}w_C(x)\left(\sum_{i=1}^n r_i(x)\|c^*_i-x\|^p+r_i(x)\|\pi(c^*_i)-\pi(x)\|^p\right)}\\[10pt]
&\leq 2\alpha\C_p(\calC^*) + (1+\eps)^p \cdot 2\alpha\C_p(\calC^*) + \alpha \C_p(\calC^*) \\[10pt]
&\le5(1+\eps)^p\alpha\,\C_p(\calC^*).
\end{align*}
By Markov's inequality, we have that
\[\PPr{\eps^{1-p}\I\{\calE\}\sum_{x\in C\setminus C'}R_x\ge c_{\eta\eps\alpha}\C_p(\calC^*)}\le\eps^{1-p}\cdot\frac{5(1+\eps)^p\alpha\C_p(\calC^*)}{c_{\eta\eps\alpha}\C_p(\calC^*)}=\eta.\]
Recalling that bounds \eqref{eq: cost(pi(C)) bound} and \eqref{eq: cost(C) bound} hold with probability at least $1 - \binom{n}{2}\d$, by a union bound we have 
\begin{align*}
\C_p(\pi(C))&\le A(\C_p(C)+c_{\eta\eps\alpha}\C_p(\calC^*))\\
\C_p(C)&\le A(\C_p(\pi(C))+c_{\eta\eps\alpha}\C_p(\calC^*)),
\end{align*}
with probability at least $1-\eta-\binom{n}{2}\delta$. 
\end{proof}

We now show that the cost of any valid solution to the Wasserstein barycenter problem (again in the sense of Definition \ref{def:solution}) is roughly preserved under a random $(\eps,\delta,\alpha)$-standard dimensionality reduction map. 
The proof is similar to Theorem 3.5 in \cite{MakarychevMR19}, but we again use a weighted version of the distortion graphs.
\begin{theorem}
\label{thm:wb:main}
Let $\mu_1,\ldots,\mu_k$ be $k$ discrete distributions with support size $n$ on $\R^d$, and let $X = \bigcup_{i=1}^k \mathrm{supp}(\mu_i)$.
Given $\eps\in(0,1/4)$ and $\delta\in(0,1)$ and $\pi:\R^d\to\R^m$ with 
\[m=O\left(\frac{\log\frac{n}{\delta}+p\log\frac{1}{\eps}+p^2}{\eps^2}\right),\]
then with probability at least $1-\delta$, we have that simultaneously for every solution $\calC = (C_1, \ldots, C_n)$ and corresponding weight functions $w_j(\cdot)$ of $X$,
\begin{align*}
\C_p(\pi(\calC))&\le(1+\eps)^{3p}\C_p(\calC)\\
(1-\eps)\C_p(\calC)&\le(1+\eps)^{3p-1}\C_p(\pi(\calC)).
\end{align*}
\end{theorem}
\begin{proof}
We define $\theta=\min(\eps^{p+1}3^{-(p+1)(p+2)},\delta\eps^p/(10n(1+\eps)^{4p-1},1/10^{p+1})$. 
Let $\delta'\le\min(\theta^7/600,\theta/n)$, $\binom{n}{2}\delta'\le\delta/2$, $\alpha\le\delta$, and $\eta\le\frac{\delta}{2}$. 
Let the constant in $m$ be sufficiently large, so that $\pi$ is a random $(\eps,\delta',\alpha)$-standard dimensionality reduction map. 
Note that the constant $m$ is independent of the quantities $\eps,\delta',\alpha$ due to the properties of a standard dimensionality reduction map. 
Then we seek to apply Theorem~\ref{thm:cluster:preserve} and note that with these values of $\eps$, $\delta'$, and $\alpha$, we have $A\le(1+\eps)^{3p-1}$ and $c_{\eta\eps\alpha}\le\eps/n$. 

Let $\calE$ be the event that we have
\begin{align}
\C_p(\pi(\calC))&\le(1+\eps)^{3p}\C_p(\calC) \label{eq:good1}\\
(1-\eps)\C_p(\calC)&\le(1+\eps)^{3p-1}\C_p(\pi(\calC)) \label{eq:good2}
\end{align}
simultaneously any valid solution $\calC$ to the Wasserstein barycenter problem. 
Suppose that the event $\calE$ does not occur, so that there exists a solution $\calC=\{C_1,\ldots,C_n\}$ and corresponding weight functions $w_j(\cdot)$ that violates \eqref{eq:good1} or \eqref{eq:good2}.
If \eqref{eq:good1} fails to hold, then
\begin{align*}
    \sum_{i=1}^n\C_p(\pi(C_i))&= \C_p(\pi(\calC)) \\
    &\geq (1 + \eps)^{3p} \C_p(\calC) \\
    &\geq ((1 + \eps)^{3p - 1} + \eps) \C_p(\calC) \\
    &\geq A\left(\sum_{i=1}^n \C_p(C_i)\right) + \eps \C_p(\calC).
\end{align*}
Similarly, if \eqref{eq:good2} fails to hold, then
\begin{align*}
    \sum_{i=1}^n\C_p(C_i) &= (1-\eps)\C_p(\calC) + \eps\C_p(\calC) \\
    &\geq (1+\eps)^{3p-1}\C_p(\pi(\calC)) + \eps\C_p(\calC) \\
    &\geq A\left(\sum_{i=1}^n \C_p(\pi(C_i))\right) + \eps\C_p(\calC).
\end{align*}
It follows that there exists some $i\in[n]$ such that at least one of the following inequalities holds:
\begin{align*}
\C_p(\pi(C_i))&\ge A\,\C_p(C_i)+\frac{\eps}{n}\,\C_p(\calC)\\
\C_p(C_i)&\ge A\,\C_p(\pi(C_i))+\frac{\eps}{n}\,\C_p(\calC).
\end{align*}
Let $\calC^*$ be the optimal solution (i.e. the actual Wasserstein barycenter) for $\mu_1,\ldots,\mu_k$. In particular, this means that $\C_p(\calC^*) \leq \C_p(\calC)$.
Then one of the following inequalities must hold:
\begin{align*}
\C_p(\pi(C_i))&\ge A\,\C_p(C_i)+\frac{\eps}{n}\,\C_p(\calC^*)\\
\C_p(C_i)&\ge A\,\C_p(\pi(C_i))+\frac{\eps}{n}\,\C_p(\calC^*).
\end{align*}
By Theorem~\ref{thm:cluster:preserve}, one of these inequalities can hold with probability at most $\eta+\binom{n}{2}\delta'$. 
Since $\eta\le\frac{\delta}{2}$ and $\binom{n}{2}\delta'\le\frac{\delta}{2}$, it follows that $\calE$ occurs with probability at least $1-\delta$. 
\end{proof}

Finally, we show how to rescale the parameters of Theorem~\ref{thm:wb:main} to prove Theorem~\ref{thm:jl:main}. 
\begin{proof}[Proof of Theorem \ref{thm:jl:main}]
Observe that Theorem~\ref{thm:wb:main} with a rescaling of $\eps':= (1+\eps)^{1/(3p)-1}=O(\eps/p)$ immediately implies the desired claim. 
\end{proof}

\textbf{Fast dimensionality reduction.}
The dimensionality reduction maps of Theorem~\ref{thm:wb:main} generally require multiplication by a dense matrix of (scaled) subgaussian random variables. 
Thus for $\delta=O(1)$ and $p=O(1)$, applying the dimensionality reduction map using rectangular matrix multiplication takes $O\left(\frac{dkn\log n}{\eps^2}\right)$ time. 
We provide a tradeoff between runtime and dimension using the following observation:
\begin{theorem}
\label{thm:fast:jl:mmr19}
\cite{MakarychevMR19}
There exists a family of $(\eps,\delta,\alpha)$-dimensionality reduction maps $\pi:\mathbb{R}^d\to\mathbb{R}^m$ with $m=O\left(\frac{p^6}{\eps^2}\log^2\frac{k}{\eps}{\delta}\right)$ with runtime $O(d\log d)$ on an input vector $v\in\mathbb{R}^d$. 
\end{theorem}
We describe the construction of $\pi$ in Theorem~\ref{thm:fast:jl:mmr19} as in \cite{AilonC09, MakarychevMR19}. 
By a standard padding with zeros argument, we first assume that $d$ is a power of two. 
We define $D$ to be a diagonal $d\times d$ matrix with i.i.d. uniform signs, i.e., $\pm1$ entries. 
We define $H$ to be a normalized Hadamard transform so that $H$ is an orthogonal matrix with all entries $\pm\frac{1}{\sqrt{d}}$ and $Hv$ can be computed in $O(d\log d)$ time. 
We define $S$ to be a diagonal ``sampling'' matrix with i.i.d. entries, so that $\PPr{S_{i,i}=\frac{\sqrt{d}}{m}}=\frac{m}{d}$ and $\PPr{S_{i,i}=0}=1-\frac{m}{d}$. 
Let $\Pi:\mathbb{R}^d\to\mathbb{R}^d$ be defined by $\Pi=SHD$ and note that the expected dimension of the image of $\Pi$ is $m$. 
Then $\pi$ is defined to be the image of $\Pi$, conditioned on the event that the dimension of the image of $\Pi$ is at most $10d$~\cite{AilonC09, MakarychevMR19}. 
Hence, we obtain the following fast dimensionality reduction:
\begin{corollary}[Fast dimensionality reduction]
Let $\mu_1,\ldots,\mu_k$ describe an instance of the Wasserstein barycenter problem with the $L_p$ objective in $\R^d$.  
Given $\eps\in(0,1/4)$ and $\delta\in(0,1)$, there exists $\pi:\R^d\to\R^m$ with 
\[m=O\left(\frac{p^6}{\eps^2}\log^2\frac{k}{\eps}{\delta}\right),\]
that uses $O(d\log d)$ runtime to apply the mapping to each point and provides the same guarantees as Theorem~\ref{thm:wb:main}. 
\end{corollary}

\textbf{Alternative proof for Theorem~\ref{thm:jl:main}.}
An anonymous NeurIPS 2021 reviewer pointed out the following alternative proof for Theorem~\ref{thm:jl:main}. 
Given the $k$ distributions $\mu_1,\ldots,\mu_k$ in $\mathbb{R}^d$, consider each distribution $\mu_i$ as a multiset $U_i$ of $\mathbb{R}^d$ of size $M$ for a sufficiently large $M$. 
Then the Wasserstein barycenter problem can be rewritten as the optimization problem
\[\min\sum_{i=1}^k\frac{1}{M}\sum_{u\in U_i}\|u-c_{f_i(u)}\|_2,\]
for an assignment function $f_i:U_i\to[n]$, subject to the constraints:
\begin{enumerate}
\item
$c_1,\ldots,c_n\in\mathbb{R}^d$
\item
$\frac{1}{M}\sum_{u\in U_i}\mathbf{1}[f_i(u)=j]=a_j$, where $\sum_{j\in[n]}a_j=1$.
\end{enumerate}
Setting $U$ to be the multi-set defined by the union of all $U_i$ with $i\in[k]$ and $f:U\to[n]$ defined by $f(u)=f_i(u)$ for $u\in U_i$, the above optimization can be further rewritten as
\[\min\frac{1}{M}\sum_{u\in U}\|u-c_{f(u)}\|_2,\]
subject to the same constraints. 
Since this is a constrained $k$-median clustering problem and \cite{MakarychevMR19} show that the cost of \emph{every} clustering is preserved to within a $(1+\eps)$-factor under a projection to $O\left(\frac{1}{\eps^2}\log n\right)$ dimensions, then the optimal clustering under the above constraints are also preserved to within a $(1+\eps)$-factor. 

\section{Proofs for Section \ref{sec:coresets}}\label{sec:coreset_proofs}

We need the following theorem which relates the size of coresets, obtained from importance sampling according to sensitivity values, to the pseudo-dimension of a related function class. 
Sensitivity sampling has been used to design coresets for many problems in machine learning such as support vector machine, Gaussian mixture models, projective clustering, principal component analysis,  $M$-estimators, Bayesian logistic regression, and generative adversarial networks, e.g., see recent surveys on coresets such as~\cite{bachem2017practical,Feldman20}.

From Theorem \ref{thm:coreset_size}, we now need to bound the following two things to obtain a coreset.
\begin{enumerate}
    \item The total sensitivity $S$,
    \item The pseudo-dimension of $\mathcal{F}$.
\end{enumerate}

(Note that Theorem \ref{thm:coreset_informal} in \cite{bachem2017practical, lang_thesis} is stated for coresets of general functions, we specialize it to the case of Wasserstein Barycenters.)

We begin by bounding $(1)$, the total sensitivity $S$. To do so, we need to define a function $s$ that informs how we sample the distributions in $M$. The following lemma shows that it suffices to use a constant factor approximation to the best barycenter solution to perform the sampling. 

\begin{lemma}\label{lem:total_sensitivity_bound}
Let $\alpha, p \ge 1$ and let $\nu'$ be an $\alpha$-approximate solution to the $p$-Wasserstein Barycenter problem for the set $M$ of distributions in $\R^d$ with support size at most $n$. That is, 
\[ \frac{1}{|M|}\sum_{\mu \in M}W(\mu, \nu')^p \le  \alpha \, \frac{1}{|M|}\sum_{\mu \in M}W(\mu, \nu^*)^p \]
where $\nu^*$ is the optimal barycenter distribution. Then the sensitivity $\sigma(\mu)$ for $\mu \in M$ defined as in Definition \ref{def:sensitivity} is bounded by
\[ \sigma(\mu) \le s(\mu) \le \frac{ \alpha 2^{p-1} \, W(\mu,\nu')^p }{\frac{1}{|M|} \sum_{\widetilde{\mu} \in M} W(\widetilde{\mu}, \nu')^p} + \alpha4^{p-1} + 4^{p-1}.\]
Furthermore, it holds that 
\[ \mathfrak{S} \le  \alpha(4^{p-1} + 2^{p-1}) + 4^{p-1}.  \]
\end{lemma}
\begin{remark}
Note that $p$ is typically $O(1)$; for example $p = 1$ or $p = 2$ are the most common choices. 
\end{remark}

Now our goal is to bound $(2)$, the pseudo-dimension of the function class $\mathcal{F}$ in Theorem \ref{thm:coreset_size}. First we relate pseudo-dimension to VC dimension of related threshold functions and then we state a result relating the VC dimension to the algorithmic complexity of computing these threshold functions.

\begin{lemma}[Pseudo-dimension to VC dimension, Lemma $10$ in \cite{Feldman_gaussians}]\label{lem:PD_to_VC}
For any $f \in \mathcal{F}$, let $B_f$ be the indicator function of the region on or below the graph of $f$, i.e., $B_f(x,y)  =  \text{sgn}(f(x)-y)$. The pseudo-dimension of $\mathcal{F}$ is equivalent to the VC-dimension of the subgraph class $B_{\mathcal{F}}=\{B_f \mid f \in \mathcal{F}\}$.
\end{lemma}

Then we need the following theorem that relates VC dimension of a function class to its computational complexity. 

\begin{lemma}[Theorem $8.14$ in \cite{anthony_bartlett_1999}]\label{lem:VC_bound}
Let $h: \R^a \times \R^b \rightarrow \{ 0,1\}$, determining the class
\[ \mathcal{H} = \{x \rightarrow h(\theta, x) : \theta \in \R^a\}. \]
Suppose that any $h$ can be computed by an algorithm that takes as input the pair $(\theta, x) \in \R^a \times \R^b$ and returns $h(\theta, x)$ after no more than $t$ of the following operations:
\begin{itemize}
    \item arithmetic operations $+, -,\times,$ and $/$ on real numbers,
    \item jumps conditioned on $>, \ge ,<, \le ,=,$ and $=$ comparisons of real numbers, and
    \item output $0,1$,
\end{itemize}
then the VC dimension of $\mathcal{H}$ is $O(a^2 t^2 + t^2 a \log a)$.
\end{lemma}

Combining the previous lemmas lets us prove the following theorem. At a high level, we are instantiating Lemma \ref{lem:VC_bound} with the complexity of \emph{computing} any function in the function class $\mathcal{F}$ defined in \ref{thm:coreset_size}. This is a similar proof strategy used in \cite{Feldman_gaussians} to control the coreset size of a different problem (coresets for Gaussian mixture models) in the proof of their Theorem $2$.  However, there is a flaw in their argument as we believe that they incorrectly apply the pseudo-dimension argument to a slightly different function class. We propose a fix in our proof below.

Finally note that the function class $\mathcal{F}$ depends on the the sampling function $s$. For our purposes, we use the sampling function defined in Lemma \ref{lem:total_sensitivity_bound} which samples according to a fixed $O(1)$ approximate solution.

\begin{theorem}\label{thm:final_PD}
Consider the set of functions $\mathcal{F}$ defined in Theorem \ref{thm:coreset_size} for $s:M \rightarrow \R$ as defined in Lemma \ref{lem:total_sensitivity_bound}. The pseudo-dimension of $\mathcal{F}$ is $O(n^8d^2)$.
\end{theorem}

Altogether, we can prove the following bound on the size of coresets for the Wasserstein Barycenter problem which states that $k$, the number of distributions, can be reduced to $\poly(n,d)$ for constant $p$.

\begin{theorem}\label{thm:final_coreset_bound}
Let $\delta, \eps \in (0,1)$. Let $\nu'$ be an $\alpha$ approximation to the $p$-Wasserstein Barycenter problem for a set $M$ of distributions in $\R^d$ with support size at most $n$ and let $s : M \rightarrow \R$ be defined as in Lemma \ref{lem:total_sensitivity_bound}. Consider sampling a subset $K \subseteq M$ of size $\widetilde{\Omega}( \alpha 4^{p-1} n^8d^4/\eps^2)$ where $\widetilde{\Omega}$ hides logarithmic factors. Then $K$ satisfies Definition \ref{def:coreset} with probability $1-\delta$.
\end{theorem}
\begin{proof}
The result follows from instantiating Theorem \ref{thm:coreset_size} with the bound of $S$ from \ref{lem:total_sensitivity_bound} and the pseudo-dimension bound in \ref{thm:final_PD}.
\end{proof}

\begin{remark}
Again we remark that we are not optimizing for the exact constants in the exponents in Theorem \ref{thm:final_coreset_bound}. There are several places where such optimizations can possibly be made. For example, using a faster algorithm than the Hungarian algorithm to argue about the pseudo-dimension bound in Theorem \ref{thm:final_PD}. However, any such optimizations would result in coresets of size $\text{poly}(n,d)$ if we are to use the sensitivity sampling framework.
\end{remark}

\begin{proof}[Proof of Lemma \ref{lem:total_sensitivity_bound}]
Let $\nu$ denote an arbitrary barycenter distribution.
For any $\mu \in M$, the triangle inequality gives us 
\[W(\mu,\nu)^p \le 2^{p-1}( W(\mu, \nu')^p + W(\nu', \nu)^p) \]
where we have used the fact that $(x+y)^p \le 2^{p-1}(x^p+y^p)$ for non-negative $x,y$ and $p \ge 1$. Using a similar reasoning, we have
\[ W(\nu', \nu)^p \le 2^{p-1}(W(\nu', \mu') + W(\mu', \nu)) \]
for every $\mu' \in M$. Averaging over all $\mu'$ gives us
\[  W(\nu', \nu)^p \le \frac{2^{p-1}}{|M|} \sum_{\mu' \in M} (W(\nu', \mu')^p + W(\mu', \nu)^p). \]
It follows that 
\begin{align*}
    \frac{W(\mu,\nu)^p}{\frac{1}{|M|} \sum_{\widetilde{\mu} \in M} W(\widetilde{\mu}, \nu)^p}  &\le \frac{2^{p-1}W(\mu,\nu')^p }{\frac{1}{|M|} \sum_{\widetilde{\mu} \in M} W(\widetilde{\mu}, \nu)^p} +  \frac{ \frac{4^{p-1}}{|M|}\sum_{\mu' \in M} \left( W(\nu', \mu')^p + W(\mu' , \nu)^p \right) }{\frac{1}{|M|} \sum_{\widetilde{\mu} \in M} W(\widetilde{\mu}, \nu)^p} \\[10pt]
    &\le  \frac{ \alpha 2^{p-1} \, W(\mu,\nu')^p }{\frac{1}{|M|} \sum_{\widetilde{\mu} \in M} W(\widetilde{\mu}, \nu')^p} + \frac{\frac{\alpha 4^{p-1}}{|M|} \sum_{\mu' \in M} W(\nu', \mu')^p}{\frac{1}{|M|} \sum_{\widetilde{\mu} \in M} W(\widetilde{\mu}, \nu')^p} + 4^{p-1} \\[10pt]
    &\le \frac{ \alpha 2^{p-1} \, W(\mu,\nu')^p }{\frac{1}{|M|} \sum_{\widetilde{\mu} \in M} W(\widetilde{\mu}, \nu')^p} + \alpha4^{p-1} + 4^{p-1} =: s(\mu)
\end{align*}
where we have used the fact that $\nu'$ is an $\alpha$-approximation to the optimal barycenter and thus,
\[ \sum_{\widetilde{\mu} \in M} W(\widetilde{\mu}, \nu)^p \ge \sum_{\widetilde{\mu} \in M} W(\widetilde{\mu}, \nu^*)^p \ge \frac{1}{\alpha} \sum_{\widetilde{\mu} \in M} W(\widetilde{\mu}, \nu')^p \]
by assumption on $\nu^*$ and $\nu'$. This gives us
\[ \mathfrak{S} \le \frac{1}{|M|} \sum_{\mu \in M} s(\mu) \le 4^{p-1}+\alpha 4^{p-1} + \frac{\alpha 2^{p-1}}{|M|} \,  \frac{\sum_{\mu \in M} W(\mu, \nu')^p}{ \frac{1}{|M|} \sum_{\widetilde{\mu} \in M} W(\widetilde{\mu}, \nu')^p}  = \alpha(4^{p-1} + 2^{p-1}) + 4^{p-1}.\]
Since $\mu$ and $\nu$ were arbitrary, the result follows.
\end{proof}

\begin{proof}[Proof of Theorem \ref{thm:final_PD}]
Let $\nu \in N$ where $N$ is the set of all possible barycenter distributions with support size $n$, as defined in Theorem \ref{thm:coreset_size} and let $r \in \R$. 
Let $M$ be a set of $k$ different distributions on $\mathbb{R}^d$, each with support size at most $n$. 
Then for $x\in M$, we define $h:N\times\mathbb{R}\times M\to\{0,1\}$ by $h(\nu,r,x)=h_{\nu,r}(x) = \I\{W(x,\nu)^p/s(x) \ge r\}$.

We remark that the conceptually similar proof of Theorem $2$ in \cite{Feldman_gaussians} used to bound the coreset sizes of Gaussian mixture models erroneously omits the function $s(x)$ in the definition of $h$ above.

Now let the corresponding function class $\mathcal{H}$ be defined as 
\[ \mathcal{H} = \{ h_{\nu, r} : M \rightarrow \{0,1\} \mid \nu \in N, r \in \R\}. \]
Note that computing $W(x,\nu)^p$ is equivalent to computing the minimum cost bipartite matching between the weighted points of $x$ and $\nu$ with edge costs coming from the Euclidean metric raised to the $p$th power. By the well known Hungarian algorithm, this can be computed in $O(n^3+n^2d)$ arithmetic steps where the first term is from the Hungarian algorithm and the second term is to compute the edge costs between $x$ and $\nu$. Furthermore, computing $s(x)$ can also be done in $O(n^3+n^2d)$ since we need to find the cost of the matching between $x$ and $\nu'$ where $\nu'$ is the approximate solution used to define $s$ in Lemma \ref{lem:total_sensitivity_bound} (the other terms of $s(x)$ are constant).

Therefore by Lemma \ref{lem:VC_bound}, the VC dimension of $\mathcal{H}$ is at most  $O( (nd)^2 \cdot (n^3+n^2d)^2) = O(n^8d^4)$ since we need $O(nd)$ variables to define $x$ and $\nu$. Now note that the function class $\mathcal{H}$ is equivalent to the function class $\mathcal{J}$ defined as 
\[ \mathcal{J} = \{ f_{\nu, r} : M \rightarrow \{0,1\} \mid \nu \in N, r \in \R\} \]
where 
\[ f_{\nu, r} =  \I\left\{\frac{W(\cdot, \nu)^p}{ s(\cdot) \, \sum_{m \in M}  W(m,\nu)^p}  \ge r \right\}\]
This is because we are letting $r$ range over all all the reals in the definition of $\mathcal{H}$. Therefore it also follows that the VC dimension of class $\mathcal{J}$ is $O(n^8d^4)$. Finally by Lemma \ref{lem:PD_to_VC}, the pseudo-dimension of $\mathcal{F}$ as defined in Theorem \ref{thm:coreset_size} can be bounded by $O(n^8d^4)$.
\end{proof}

\section{Lower Bound Proofs for Section \ref{sec:lower bounds}} \label{sec:supp_lower_bounds}


We now turn to proving lower bounds, showing that our dimensionality reduction is optimal up to constant factors. To begin, we first need two auxiliary results regarding random linear transformations.

The proof of Theorem 9 in \cite{KaneMN11} states the following:
\begin{theorem}
\label{thm:gauss:distort}
\cite{KaneMN11}
Let $M:\mathbb{R}^d\to\mathbb{R}^m$ be a linear transformation with $d>2m$ and $\eps>0$ sufficiently small. 
Then for a randomly chosen unit vector $u\in\mathbb{R}^d$
\[\PPr{\|Mu\|^2<1-\eps}\ge\exp(-O(m\eps^2+1))).\]
\end{theorem}
By applying rotational invariance of a standard Gaussian, we arrive at the following corollary of Theorem \ref{thm:gauss:distort}.
\begin{corollary}
\label{cor:gauss:distort}
Let $M:\mathbb{R}^d\to\mathbb{R}^m$ be a random matrix with i.i.d. entries from $\mathcal{N}(0,\sigma)$, $d>2m$ and $\eps>0$ sufficiently small. Then for any vector $u\in\mathbb{R}^d$,
\[\PPr{\|Mu\|^2<1-\eps}\ge\exp(-O(m\eps^2+1))).\]
\end{corollary}

We are now equipped to prove Theorem~\ref{thm:wb:lb}, restated here for convenience.

\thmwblb*
\begin{proof}
Let $t>0$ be a parameter and $N>0$ be a sufficiently large constant. 
Consider the points $S=p_1,p_2,\ldots,p_t,q_1,q_2,\ldots,q_t$ so that $p_i=Ne_i$ and $q_i=(N+1)e_i$ for each $i\in[t-1]$, where $e_i$ is the $i$-th elementary vector. 
Let $p_t=Ne_t$ and $q_t=(N+1-C\eps)e_t$ for a parameter $C>0$. 
Hence we have:
\begin{enumerate}
\item $\|p_i-q_i\|=1$ for each $i\in[t-1]$. 
\item $\|p_t-q_t\|=1-C\eps$.
\item $\min_{i\neq j}(\|p_i-p_j\|,\|p_i-q_j\|,\|q_i-q_j\|)\ge N(1-C\eps)\sqrt{2}$.
\end{enumerate}

Let $n=2t-1$ and consider the $k=2t$ distributions $\mu_1,\ldots,\mu_{2t}$ so that for each $i\in[t]$, $\mu_i$ has weight $\frac{1}{2t-1}$ on each of the points in $S$ except $p_i$, at which it has weight zero. 
Similarly, for each $i\in[t]$, suppose $\mu_{i+t}$ has weight $\frac{1}{2t-1}$ on each of the points in $S$ except $q_i$, at which it has weight zero. 
Thus, the total weight across all distributions at each of the $2t$ points in $S$ is exactly $1$. 
It can easily be shown that the barycenter of support size at most $n$ has cost $(1-C\eps)^p$, by choosing the points $p_i$ and $q_i$ for each $i\in[t-1]$ and then either $p_t$ or $q_t$. 

We now show that for a Gaussian matrix $M$ with dimension $m=\left(\frac{\log n}{1000\eps^2}\right)$, with high probability there exists some $j\in[t-1]$ such that $\|Mp_j-Mq_j\|\le\|Mp_t-Mq_t\|$. 
First note that $\|Mp_t-Mq_t\|$ equals $(1-C\eps)$ times a random variable that follows a Chi-squared distribution with $m$ degrees of freedom. 
By standard concentration inequalities on the sum of $m$ independent $\chi^2$ variables, e.g., Equation 2.21 in~\cite{wainwright2019high}, we have that
\[\PPr{\chi^2_m\le(1-\eps)m}\le\exp\left(-\frac{m\eps^2}{8}\right)\le0.01.\]
Thus the probability that $\|Mp_t-Mq_t\|\le(1-C\eps)(1-\eps)$ is at most $0.01$.

Moreover we have that for $i\neq[t]$, $p_i-q_i$ is a unit vector in $\mathbb{R}^d$. 
Thus by Corollary~\ref{cor:gauss:distort}, we have that 
\[\PPr{\|Mp_i-Mq_i\|^2<1-(4C)\eps}\ge\exp(-O(m(4C)\eps^2+1))\ge\frac{1}{n^{1/5}},\]
for $m=\left(\frac{\log n}{1000C^2\eps^2}\right)$.
Therefore since $t=\Omega(n)$, we have that there exists $i\neq[t-1]$ with $\|Mp_i-Mq_i\|^2<1-(4C)\eps$ with probability at least $0.99$. 

Hence with probability at least $0.98$, the optimal clustering in the projected space will be the projection of the points $p_1,\ldots,p_{j-1},p_{j+1},\ldots,p_t$, the points $q_1,\ldots,q_{j-1},q_{j+1},\ldots,q_t$, and either the point $p_j$ or $q_j$. 
Thus in this case, the corresponding cost in the original space is exactly $1$, so that the dimension reduction map does not allow a $(1+\eps)$-approximation to the optimal clustering.
\end{proof}

Next, we turn to proving lower bounds on dimensionality reduction for the optimal transport problem. Again, we first need an auxiliary concentration result for high-dimensional Gaussians.

\begin{lemma} \label{lem:nearby}
    Let $C \ge 1$ and fix some point $v$ of norm at most $C$ in $\mathbb{R}^d$. Then, if $x \sim \frac{1}{\sqrt{d}} \cdot \mathcal{N}(0, I_d)$ is a $d$-dimensional scaled multivariate Normal, then $\Pr(\|x-v\| \le \frac{1}{C}) \ge n^{-1/10},$ if $d \le \log n/(10 C^2)$ and $n$ is sufficiently large.
\end{lemma}

\begin{proof}[Proof of Lemma \ref{lem:nearby}]
    By the rotational symmetry of the multivariate normal, assume $v = (r, 0, \dots, 0),$ where $0 \le r \le C.$ Then, if $x = (x_1, y)$ for $x_1 \in \mathbb{R}, y \in \mathbb{R}^{d-1}$, then if $r-\frac{1}{2C} \le x_1 \le r$ and $\|y\| \le \frac{1}{2C},$ then we indeed have $\|x-v\| \le \frac{1}{C}.$ Since $\sqrt{d} x_1 \sim \mathcal{N}(0, 1)$ and $r \le C$, the probability that $r - \frac{1}{2C} \le x_1 \le r$ equals the probability that $\mathcal{N}(0, 1) \in [(r-1/2C) \sqrt{d}, r \sqrt{d}],$ which is at least $\frac{\sqrt{d}}{2C} \cdot \frac{1}{\sqrt{2 \pi}} \cdot e^{-C^2 d/2}.$ Moreover, the probability that $\|y\| \le \frac{1}{2C}$ is at least $\left(\frac{1}{2eC}\right)^d$. Therefore,
\begin{align*}
\Pr\left(\|x-v\| \le \frac{1}{C}\right) &\ge \frac{\sqrt{d}}{2C} \cdot \frac{1}{\sqrt{2 \pi}} \cdot e^{-C^2 d/2} \cdot \left(\frac{1}{2 e C}\right)^d \\
&\ge n^{-1/10},
\end{align*}
    where the last inequality is true because $d \le \log n/(10 C^2)$ and that $n$ is sufficiently large.
\end{proof}

\begin{proof}[Proof of Theorem \ref{thm:bad reduction}]
Our point sets will be $A \cup B = \{e_i\} \cup \{e_i/2\}$ with the property that $e_i$ and $e_i/2$ will be in different sets and we will alternate the $i$'s such that $e_i \in A$. The optimal matching in $\mathbb{R}^d$ is to match each $e_i$ to $e_i/2$ leading to cost $d/2$. 

Our strategy is to show that if we project $A \cup B$ to $m = o(\log d)$ dimensions, then we can find a matching of cost $o(d)$. Towards that end, let $C = \sqrt{\log d/10m} = \omega(1)$ and let $\pi$ be a random projection to $o(\log d)$ dimensions. First, we will show that points $\pi e_i$ with $\| \pi e_i\| \le C$ will have `many' other points $\pi e_j$ sufficiently near by so that we can match $\pi e_i$ to $\pi e_j$ (assuming they are in different sets). We then show that the points with $\| \pi e_i \| \ge C$ can be disregarded.

More formally, by Lemma \ref{lem:nearby}, the number of other points $e_j$ such that $\| \pi e_i - \pi e_j \| \le 1/C$ and $e_j$ is in a different set than $e_i$ is a binomial random variable $B(d-1,q)$ where $q \ge d^{-1/10}/2$. Therefore the number of such $j$'s is at least $d^{c'}$ for some constant $c' > 0$ except with probability at most $\ll 1/d$. By a union bound, we can assume that every $i$ such that $\|\pi e_i \| \le C$ has at least $d^{c'}$ other $\pi e_j$'s such that $\| \pi e_i - \pi e_j\| \le 1/C$ and $e_j$ is in a different set than $e_i$. Now consider the following greedy matching procedure to match the points $e_i$ with $\|\pi e_j \le C\|$ which may not be optimal: for every such $\pi e_i$, we try to match it to any $\pi e_j$ that is within distance $1/C$ greedily (we also map $\pi e_i/2$ to $\pi e_j/2$). We do this until it is no longer possible. Then, we try to match each $\pi e_i$ to some $\pi e_j$ within distance $2/C$ greedily until no longer possible. Then, we just match $\pi e_i$ to $\pi e_i/2$. Note that every possible match contributes $O(1/C)$ to the matching cost so altogether, this greedy matching has cost at most $O(d/C) = o(d)$. 

We now want to show that not many of the $\pi e_i$ will be leftover that have to be matched to $\pi e_i/2$. Consider maximally covering the set of all such $\pi e_i$ that have to be matched to $\pi e_i/2$ with disjoint balls of radius $1/C$. First, every such $ \pi e_i$ must be in some ball since other wise, it would have been within radius $2/C$ of some $\pi e_{i'}$ and we would have matched them. Now each ball intersects with at least $d^{c'}$ other points in $A \cup B$ by our calculation in the previous paragraph. Therefore, there can be at most $O(d^{1-c'})$ such balls and hence, the matching cost induced by these points is at most $O(C d^{1-c'}) = o(d)$ as well.

Now we just have to deal with points $e_i$ that satisfy $\|\pi e_i\| \ge C$. If they are not matched already, we just match them to $\pi e_i/2$. The expected cost incurred by one of these edges in the matching is 
\begin{align*}
\E\left[\|Ge_i\| \cdot \mathbf{1}_{\|Ge_i\| \ge C}\right] &\le \sqrt{\E\left[\|Ge_i\|^2\right] \cdot \Pr(\|Ge_i\| \ge C)} \\
&\le \sqrt{1 \cdot \exp\left(-m \cdot (C-1)^2/8\right)} \\
&\le \exp\left(-(C-1)^2/16\right) \le \frac{1}{C},
\end{align*}
so the total expected cost from these edges is at most $O(d/C) = o(d)$. Finally by an application of Markov's inequality and a union bound, we have that with probability at least $2/3$, we can find a matching in $\mathbb{R}^m$ with cost at most $o(d)$ and hence, the optimal matching in the projected space has cost at most $o(d)$, as desired.
\end{proof}

If $x \in S^{m-1}$ and $\pi$ is an appropriately normalized Gaussian dimensionality reduction map, then the following statements hold about the distribution of $\|\pi x\|$ \cite{indyknaor}:
\begin{align}
    \Pr( | \|\pi x\| -1 | \ge t) &\le \exp(-dt^2/8), \label{eq:Gxa} \\
    \Pr( \| \pi x\| \le 1/t) &\le \left( \frac{3}t \right)^d. \label{eq:Gxb}
\end{align}

Finally, we prove Theorem \ref{thm:bad pullback}. Note that Theorem \ref{thm:bad reduction} states that after we perform a random projection to $o(\log n)$ dimensions, the \emph{cost} (i.e., the actual objective numerical value) of the optimal matching in the projected space will be much smaller than the cost of the optimal matching in the original dimension. This highlights that if we just wish to approximate the cost of the matching, we cannot do better than the standard JL lemma dimension bound. Note that given Theorem \ref{thm:bad pullback} it is still possible that the optimal matching in the projected dimension is approximately equal to the optimal matching in the original dimension since Theorem \ref{thm:bad pullback} is only addressing the cost. We show in the proof of Theorem  \ref{thm:bad pullback} that this is not the case; the optimal matching in the projected space will induce a poor matching in the original dimension if we project to much fewer than $\log n$ dimensions.

\begin{proof}[Proof of Theorem \ref{thm:bad pullback}] Many details of this proof follow similarly as in the proof of Theorem \ref{thm:bad reduction}. Let $C = \sqrt{\log d/10m} = \omega(1)$. Our point sets will be $A \cup B = \{e_i \cdot k/C \}$ for all $1 \le i \le d$ and $1 \le k \le C$. We refer to $A$ and $B$ as ``classes'' and assume that $C$ is an even integer. The partition of the points is as follows. For a fixed $i$, the points $e_i \cdot k/C$ will alternate which set they belong to, i.e, $e_i/C $ will be in $A$, $2e_i/C$ will be in $B$ etc. We will also impose the condition that half of the $e_i$'s will be in $A$ and the other half will be in $B$. Now note that the optimal matching in $\R^d$ is to just match each $e_i \cdot k/C$ to $e_i \cdot (k+1)/C$ for $1 \le k \le C-1$ which results in matching cost $O(d)$.

Now consider a random projection $\pi$ to $m = o(\log d)$ dimensions. Our strategy is to show that the optimal matching in $\R^m$ will contain many edges between different $e_i$'s which will induce a large matching cost in $\R^d$. 

Towards that end, define `level $k$' to be the set of points of the form $e_i \cdot k/C$ for some $i$. First note that if $\pi$ is a Gaussian random projection, we have that $\|\pi e_i\| \in  [1/10, 100]$ with probability at least $1 - \exp(-m/10) - (3/100)^m  > 0.06$ from equations \eqref{eq:Gxa} and \eqref{eq:Gxb}. Thus we can say by a standard Chernoff bound that a $\Theta(1)$ fraction of $e_i$ will satisfy $\|\pi e_i\| = \Theta(1)$ with exponentially small failure probability. By Lemma \ref{lem:nearby}, for each such $\pi e_i$ , there exists some $e_j$ such that $\| \pi e_i - \pi e_j \| \le 1/(100C)$ (again up to some exponentially small failure probability). Since the basis vectors are equally partitioned into the two classes, we can further assume that $e_j$ is in a different class than $e_i$.

Let $I$ be the set of $i$'s such that $\|\pi e_i\| = \Theta(1)$ and there is some $j$ such that $\| \pi e_i - \pi e_j \| \le 1/(100C)$ and $j$ is in a different class. For each $i \in I,$ the distance between $ \pi e_i \cdot k/C$ and $ \pi e_i \cdot \ell/C$ for any $\ell \neq k$ is at least $\frac{1}{10C}$ but the distance between $\pi e_i \cdot k/C$ and $ \pi e_j \cdot k/C$ is at most $\frac{1}{100C}$. Thus at all levels, we can potentially switch the matching between $ \pi e_i \cdot k/C$ and $ \pi e_i \cdot (k+1)/C$ or $ \pi e_i \cdot (k-1)/C$ (if it exists) to $ \pi e_i \cdot k/C$ and $ \pi e_j \cdot k/C$ and the same for the point that $\pi e_i \cdot k/C$ was matched to.
Therefore, almost all except possibly $1$ of the indices in $I$ across all levels will be matched to a point that comes from a different basis vector. Thus the pullback cost is at least some absolute constant times
\[\sum_{i \in I} \sum_{k = 1}^{C} \frac{k}{C} \ge \frac{C}{2} \cdot |I|,\]
which is least $\Omega(C \cdot d) = \Omega(C \cdot M) = \omega(M)$, as desired.
\end{proof}

\section{Connections to Constrained Low-Rank Approximation}\label{sec:sup_constrained}
\cite{CohenEMMP15} previously showed that the problem of $k$-means clustering can be formulated as a problem of constrained low-rank approximation, a class of problems which also includes the singular value decomposition (SVD). 
In this section, we show that the problem of computing a Wasserstein barycenter can be also formulated as a problem of constrained low-rank approximation. 
Thus efficient subroutines that improve the performance of low-rank approximation can also be used to improve the performance of computing a Wasserstein barycenter. 

Recall that for an input matrix $\bA\in\mathbb{R}^{a\times b}$ and any set $S$ of rank $c$ orthogonal projection matrices in $\mathbb{R}^{a\times a}$, the goal of constrained low-rank approximation is to find
\[\bP^*=\argmin_{\bP\in S}\|\bA-\bP\bA\|_F^2.\]
\begin{proof}[Proof of Theorem \ref{thm:low_rank}]
For each point $x$, let $w_i(x)$ be the weight of $x$ in distribution $w_i$ and for each $j\in[n]$, let $w_{i,j}(x)$ be the weight of $x$ in distribution $w_i$ that is assigned to barycenter $j$, so that we have $\sum_{j\in[n]} w_{i,j}(x)=w_i(x)$ and $\sum_x w_i(x)=1$ for all $i$. 
Thus we have the Wasserstein barycenter objective as minimizing
\[\sum_{j\in[n]}\sum_{i\in[k]}w_{i,j}(x)\|x-C_j\|^p.\]
Rewriting the points of $\mu_i$ as $x_{i,1},x_{i,2},\ldots,x_{i,n}$, then the Wasserstein barycenter objective for $p=2$ is
\[\min\sum_{j\in[n]}\sum_{i\in[k]}w_{i,j}(x_{i,j})\|x_{i,j}-C_j\|^2.\]
Thus we can refold the points $x_{i,j}$ into a matrix of size $\bA\in\mathbb{R}^{nk\times d}$ so that the first row of $\bA$ consists of the $d$ coordinates of $x_{1,1}$ and more generally row $(i-1)n+j$ of $\bA$ consists of the $d$ coordinates of $x_{i,j}$. 

Suppose without loss of generality that there exists an integer $N$ such that $w_{i,j}$ is a multiple of $1/N$ for each $i\in[k],j\in[n]$. 
Let $\bB\in\mathbb{R}^{Nk\times d}$ so that each row $(i-1)n+j$ of $\bA$ consecutively appears $w_{i,j}(x_{i,j})$ times in $\bB$. 
Thus $\bB$ is essentially the matrix whose rows encode each point of each distribution, effectively duplicating each point a number of times equal to its weight in the distribution. 

We define a clustering $C=\{C_1,\ldots,C_n\}$ so that there exist weights $w_1,\ldots,w_n$ with $\sum_{j\in[n]} w_j=1$ with the property that for each $i\in[k]$ and $j\in[n]$, there are exactly $w_j\,N$ points between rows $(i-1)N+1$ and $iN$ inclusive are assigned to cluster $j$. 
Intuitively, this corresponds to each barycenter being assigned weight $w_j$ from each distribution. 
For each $j\in[n]$, let $\sigma_j$ be the centroid of all the $w_j\,Nk$ points assigned to $C_j$ and for each $r\in[Nk]$, let $C(r)\in[n]$ be the cluster to which row $r$ is assigned. 

Given a clustering $C=\{C_1,\ldots,C_n\}$, we define the cluster indicator matrix $\bX_C\in\mathbb{R}^{Nk\times n}$ to be matrix such that row $(i-1)n+j$ in $\bX_C$ has entry $\frac{1}{\sqrt{|C_\ell|}}$ in column $\ell\in[n]$ if and only if the corresponding $\frac{1}{N}$ weight of $x_{i,j}$ is assigned to cluster $C_\ell$ (and entry zero otherwise). 
Thus there exist weights $w_1,\ldots,w_n$ with $\sum_{j\in[n]} w_j=1$ such that for each $i\in[k]$ and $j\in[n]$, column $j$ has exactly $w_j\,N$ nonzero entries between rows $(i-1)N+1$ and $iN$ inclusive. 
Note in this interpretation, we further have $|C_{\ell}|=w_j\,N$. 

Since the columns of $\bX_C$ have disjoint support, then the corresponding vectors are orthonormal. 
Thus $\bX_C\bX_C^\top$ is a rank $n$ projection matrix and we can write the problem of Wasserstein barycenter as the constrained low-rank approximation
\[\min_{C}\frac{1}{N}\|\bA-\bX_C\bX_C^\top\bA\|_F^2=\sum_{r\in[Nk]}\|\bB_r-\sigma_{C(r)}\|^2.\]
Note that the cluster indicator matrix $\bX_C$ is constrained to the set of valid clusters $C$ consistent with assignments of the support points in the Wasserstein barycenter to each distribution. 
\end{proof}

\section{NP Hardness of Approximation of Wasserstein Barycenters}\label{sec:np_hard}
In this section, we show the NP-hardness of finding a Wasserstein barycenter with cost within a multiplicative 1.0013 factor of the cost induced by an optimal Wasserstein barycenter. 
We first the following statement about the hardness of approximation for $k$-means clustering:
\begin{theorem}
\label{thm:kmeans:apx}
\cite{LeeSW17}
It is NP-hard to approximate $k$-means clustering within a multiplicative factor of $1.0013$.
\end{theorem}
The proof of Theorem~\ref{thm:kmeans:apx} relies on a reduction from the Vertex Cover problem on 4-regular graphs. 
Namely, \cite{ChlebikC06} showed that it is NP-hard to distinguish whether a 4-regular graph $G$ with $n$ vertices has vertex cover size at least $A_{\max} n$ or vertex cover at most $A_{\min} n$, for some absolute constants $A_{\min}<A_{\max}$. 
\cite{LeeSW17} transformed a 4-regular graph $G$ into a graph $G'$ and embedded $G'$ into $\mathbb{R}^{3n}$ so that for the optimal $k$-means clustering cost of $G'$ (where $k$ is a function of $n$) is at least $C_{\max}$ if the smallest vertex cover of $G$ has size at least $A_{\max}$ and at least $C_{\min}$ if the smallest vertex cover of $G$ has size at most $A_{\min}$. 
As it turns out, $C_{\max}/C_{\min}=1.0013$, which shows the NP-hardness of approximating the optimal $k$-means clustering cost within a factor of $1.0013$. 

Given a set $G'$ of $N$ points in $\mathbb{R}^{3n}$, let $\mu$ be a uniform distribution on the $N$ points in $G'$ such that each point $p\in G'$ has weight $\frac{1}{N}$. 
Suppose we restrict the barycenter to have support $k$, where $k$ is the number of centers in the above $k$-means clustering instance. 
Then a set of $k$ centers $c_1,\ldots,c_k$ inducing clusters $C_1,\ldots,C_k$ on $G'$ that achieves cost $C$ for $k$-means clustering on $G'$ translates to a barycenter of support size $k$ that induces optimal transport cost $\frac{C}{N}$, where the weight of $c_i$ in the barycenter is $\frac{|C_i|}{N}$, for each $i\in[k]$. 

Thus the optimal $k$-means clustering on $G'$ has cost $C$ if and only if the Wasserstein barycenter has cost $\frac{C}{N}$. 
Therefore, we immediately have the proof of Theorem \ref{thm:wb:apx}.

\end{document}